\documentclass[12pt,letterpaper]{article}
\usepackage[utf8]{inputenc}

\title{The Role of Precedence Order in Matching\\with Multi-Criteria Admissions
\footnote{Authors are listed in alphabetical order. We are grateful to Fuhito Kojima, Kenzo Imamura, and Daisuke Hirata, as well as to all educators and government officials with whom we discussed efforts to improve the matching process for public high school admissions in Japan. This work has been supported by JST PRESTO Grant Number JPMJPR2368 and JST ERATO Grant Number JPMJER2301, Japan. All remaining errors are our own.}}
\author{Shunya Noda\thanks{Corresponding Author. Graduate School of Economics, the University of Tokyo, 7-3-1
Hongo, Bunkyo-ku, Tokyo, 113-0033, Japan. E-mail:
\href{mailto:shunya.noda@e.u-tokyo.ac.jp}{shunya.noda@e.u-tokyo.ac.jp}}
\and Ayano Yago\thanks{Department of Economics, New York University, 19 West 4th Street, New York, NY 10012, U.S.A. E-mail: \href{mailto:ayano.yago@nyu.edu}{ayano.yago@nyu.edu}}}
\date{\today}

\usepackage[utf8]{inputenc}
\usepackage[margin=1in]{geometry}
\usepackage{mathtools}
\usepackage{amsmath}
\usepackage{amsthm}
\usepackage{amssymb}
\usepackage{setspace}
\usepackage{booktabs}
\usepackage{tabularx}
\usepackage{amsmath} 
\usepackage{amssymb}
\usepackage{bm}
\usepackage{ascmac}
\usepackage{setspace}
\usepackage[at]{easylist}
\usepackage[normalem]{ulem}
\usepackage{comment}

\mathtoolsset{showonlyrefs}  

\makeatletter

\usepackage{amsthm}
\usepackage{amsfonts}
\usepackage{bbm}
\usepackage{graphics}
\usepackage{enumerate}
\usepackage{float}
\usepackage{caption}
\usepackage{subcaption}
\usepackage{natbib}
\usepackage{epigraph}
\usepackage{tabularx}
\usepackage{threeparttable}
\usepackage{array}
\usepackage{ragged2e}

\theoremstyle{definition}
\newtheorem{thm}{Theorem}
\newtheorem{lem}{Lemma}

\newtheorem{prop}{Proposition}
\newtheorem{cor}{Corollary}

\newtheorem{example}{Example}

\newtheorem{condition}{Condition}

\theoremstyle{remark}

\usepackage[hyphens]{url}
\usepackage[colorlinks,urlcolor=blue, citecolor=blue, menucolor=blue]{hyperref}

\theoremstyle{plain}

\providecommand{\keywords}[1]
{
  \small	
  \textbf{Keywords:} #1
}

\usepackage{algorithm,algpseudocode}

\usepackage{amsmath}

\newcommand*{\bmu}{\boldsymbol{\mu}}

\newcommand*{\best}{\mathbf{Best}}
\newcommand*{\topacc}{\mathbf{Top}}
\newcommand*{\swap}{\mathbf{Swap}}
\newcommand*{\relab}{\mathbf{Relab}}

\newcommand*{\Ch}{\mathrm{Ch}}

\usepackage{circledsteps}

\allowdisplaybreaks

\makeatother

\begin{document}
\maketitle

\begin{abstract}
Admission systems often fill seats by applying criterion-specific rankings
sequentially. We derive criterion-wise comparative statics with respect to
precedence order when rankings may be arbitrarily nonaligned, and each
criterion evaluates the entire admitted set through a responsive preference.
For a fixed applicant set, reversing the final two blocks makes each criterion
weakly prefer the outcome in which its block is later. This comparison can
reverse if two blocks are followed by another criterion, even under reserve-type rankings. Under
student-proposing deferred acceptance, it can also reverse with two criteria
because rejection chains change the focal college's applicant set; such
feedback is the only possible source of failure. In regular large markets, failures vanish for almost all colleges and,
under sufficient thickness, uniformly across colleges. Under the Boston mechanism, final acceptance prevents feedback; thus
the comparison holds in every finite market for fixed submitted rank-order
lists and criterion-independent acceptability.
\end{abstract}


\pagebreak

\section{Introduction}
College admissions are best understood not as a contest to identify the strongest applicants, but as a choice problem to shape the composition of incoming classes. Diversity in backgrounds, perspectives, and skills can enhance learning through interaction, improve problem-solving by bringing complementary viewpoints, and foster creativity and innovation. These considerations are not specific to higher education. Similar arguments are commonly invoked in firms and other organizations, where team composition affects their productivity. In environments characterized by complex tasks and knowledge production, heterogeneity among members is often viewed as a productive asset rather than a constraint.

To promote such diversity, universities often adopt multiple evaluation criteria rather than a single measure of merit. While this process can be theoretically formulated as a problem of selecting an optimal subset of admitted students from a pool of applicants, such a formulation is typically too complex to implement in practice. Instead, admissions systems often rely on several salient evaluation dimensions, such as academic excellence, creativity, extracurricular involvement, athletic ability, leadership, and life circumstances.\footnote{According to \citet{SFFA_Harvard_2019}, Harvard University evaluates applicants along multiple dimensions, including academic achievement, extracurricular activities, and personal qualities, and makes final admissions decisions through committee deliberation rather than mechanical aggregation of scores.}

Despite the practical importance of admission systems with multiple evaluation criteria, relatively little is known about how such systems should be designed, beyond cases for reserve or quota policies for affirmative action for targeted groups \citep[e.g.,][]{hafalir2013effective,dur2018reserve,dur2020explicit,papai2022targeted}. Even in a simple environment where applicants are evaluated along two dimensions, such as academic achievement and athletic ability, admission rules can differ. For instance, an admissions system that allocates half of the available slots based on academic performance and the other half based on athletic performance can generate an admission outcome that resembles the one produced by a single rank-order list. However, from the college's perspective, what matters is not how academically strong the admitted students are within the academic category alone, but how academically strong the entire admitted class is across all categories. In particular, applicants who meet the admission standards in both dimensions may be admitted through different criteria depending on the admissions procedure, and this choice can affect the overall composition of the class. These issues already arise in a single-college choice problem, and they become substantially more complex once interactions across multiple colleges are taken into account.

In this paper, we study how the composition of an admitted class changes when
a college varies the \emph{precedence order} in which multiple criteria are
applied. Each criterion induces a distinct ranking of applicants, and we do not impose the structural assumptions commonly adopted in studies of affirmative action (such as common underlying rankings shared across criteria). We first
analyze the college's sequential choice from a fixed applicant set. We then
embed this choice problem in matching markets under student-proposing deferred
acceptance (DA) and compare it with the Boston mechanism. This comparison
isolates two forces: the direct effect of sequential selection within a college and the feedback generated when its rejection decisions alter
applications elsewhere in the market. Our main results are summarized in Table~\ref{tab:order_results}.

\begin{table}[t]
\centering
\begin{threeparttable}
\caption{Summary of Precedence-Order Comparisons}
\label{tab:order_results}

\small
\setlength{\tabcolsep}{5pt}
\renewcommand{\arraystretch}{1.12}

\begin{tabularx}{\textwidth}{
  @{}
  >{\RaggedRight\arraybackslash}p{0.25\textwidth}
  >{\RaggedRight\arraybackslash}p{0.2\textwidth}
  >{\RaggedRight\arraybackslash}X
  @{}
}
\toprule
Setting & Reversal & Result \\
\midrule

Single-college choice
& Terminal
& Later-block advantage \\

Single-college choice
& Upstream
& No uniform comparison \\

\addlinespace[0.15em]

DA, finite market
& Terminal
& No uniform comparison \\

DA, large market\tnote{a}
& Terminal
& Later-block advantage (asymptotic) \\

\addlinespace[0.15em]

Boston, finite market\tnote{b}
& Terminal
& Later-block advantage (exact) \\

\bottomrule
\end{tabularx}

\begin{tablenotes}[flushleft]
\footnotesize

\item[]
\emph{Later-block advantage}: each criterion weakly prefers
having its own block processed later.

\item[a]
Under regularity, failures occur at a vanishing expected share of
colleges; under sufficient thickness, the comparison holds with
probability approaching one for any specified sequence of colleges.

\item[b]
Submitted rank-order lists are fixed, and acceptability at the focal
college is criterion-independent.

\end{tablenotes}
\end{threeparttable}
\end{table}

In a single-college choice problem, when admissions across multiple criteria are conducted sequentially according to a fixed order, the reserve design literature has shown that the criteria positioned later are more strongly respected. The intuition is that applicants who satisfy the admission standards of multiple criteria tend to be admitted by the criterion positioned earlier, consuming that criterion's capacity. As a result, criteria positioned later can allocate their capacity to applicants who are not preferred by other criteria but are uniquely favored by them. We show that this insight continues to hold in our model, which imposes no structural restrictions on criterion preferences, when there are two criteria or when the evaluation order differs only through a local swap at the end of the precedence order.

However, this conclusion need not extend to upstream swaps, even under
reserve-type rankings \citep[studied by, e.g.,][]{dur2020explicit,
sonmez2022affirmative,papai2022targeted}. A counterexample arises already
with three criteria whose rankings are generated from a common underlying
ranking, so that all criteria rank any two applicants within the same group
identically. More generally, whenever the two swapped criterion blocks are
followed by at least one additional criterion, the swap can change which
applicants are selected in those blocks and thereby alter the pool available
downstream. The resulting downstream selections can reverse the comparison
at the level of the final admitted set, so the criterion whose block is moved
earlier may ultimately prefer the outcome. This observation shows that,
with three or more criteria, an upstream swap need not benefit the criterion
placed later. The positive comparison therefore depends critically on the
swapped blocks being terminal.


We further show that even with two criteria, the preferences of the later-positioned criterion need not be more strongly respected in markets with multiple colleges where matching outcomes are determined by DA. Changing the precedence order at a given college can alter which applicants are rejected at that college. In a matching market with two criteria, rejected students subsequently apply to other colleges, potentially triggering additional rejections, which in turn induce further applications. This phenomenon is called a rejection chain in the literature \citep{kojima2009incentives,ashlagi2017unbalanced}. When such a rejection chain eventually returns to the original college, students who would not have applied under the original precedence order may enter its choice set. As a consequence, a criterion that is moved to an earlier position in the precedence order may be able to admit applicants it prefers more strongly. Accordingly, in matching markets, the effect of precedence order can favor an earlier-positioned criterion even when there are only two criteria.

The logic above implies that a criterion can benefit from being moved earlier
only if the order change generates a rejection chain that returns to the
originating college. In large markets where such feedback does not arise, the outcome mirrors that of the single-college setting: with two criteria, the later-positioned criterion is more strongly respected. The probability that a rejection chain returns to its originating college declines as market size grows, a phenomenon that has been extensively analyzed in the literature on DA \citep{kojima2009incentives,ashlagi2017unbalanced}. Building on this no-return logic, we show that, in large markets, positioning a criterion later leads to stronger alignment between the admissions outcome and that criterion with probability approaching one.

Finally, we show that a large-market limit is not essential
when we use the Boston mechanism, which itself prevents such feedback. Holding submitted rank-order
lists fixed and imposing criterion-independent acceptability at the focal
college, each criterion weakly prefers the order in which it is placed later
in every finite market. Before the first round in which the two precedence
orders can select different students, they generate the same acceptance and
rejection history. That first order-sensitive round fills all remaining seats,
and Boston acceptances are final, so subsequent rejections elsewhere cannot
alter the focal college's class. Because the Boston mechanism is not
strategy-proof for students \citep{abdulkadiroglu2003school,ergin2006games}, this result is pointwise in submitted rank-order
lists rather than a comparison of equilibrium reports. Together, the DA and
Boston results identify feedback before the focal college's admissions
become final as the central mechanism-design distinction.

Our analysis has direct implications for the design of admission systems that
apply multiple evaluation criteria sequentially. Public high school admissions
in Japan provide a particularly relevant institutional setting and have
recently become the subject of broader policy review, including proposals to
introduce coordinated matching mechanisms such as DA. Admissions are
administered largely at the prefectural level and typically combine a common
academic examination with junior-high-school records and, in many schools,
interviews, practical tests, or other school-specific materials. 

Sequential or multi-stage selection is far from exceptional. Hokkaido
provides a transparent example: about 70\% of capacity is first
filled under a procedure that treats examination scores and junior-high-school
records equally; two subsequent blocks of about 15\% each emphasize school
records and examination scores, respectively, and each school chooses the
order of these two blocks. Official admissions rules in Gunma, Miyagi, and
Aomori likewise specify school- or program-specific capacity shares and the
order in which distinct selection methods are applied.\footnote{See the official admissions materials issued by the boards of
education of
\href{https://www.dokyoi.pref.hokkaido.lg.jp/hk/gks/201495.html}
{Hokkaido},
\href{https://www.pref.gunma.jp/site/kyouiku/714649.html}
{Gunma},
\href{https://www.pref.miyagi.jp/site/sub-jigyou/kyo-r9-senbatsuhouhoutou.html}
{Miyagi}, and
\href{https://www.pref.aomori.lg.jp/soshiki/kyoiku/e-gakyo/R09motomeru.html}
{Aomori}.}
Related multi-stage selection procedures are also found in many other
prefectures across Japan.
Although the institutional details differ, these systems make the sequencing
of evaluation criteria a substantive component of admissions design.

Osaka Prefecture provides a particularly salient reform case because it is
changing where a school-specific criterion enters the selection sequence.
In addition to exam scores
and grade records, Osaka allows schools to use
school-specific criteria reflecting their educational missions. Through the
2027 admission cycle, applicants in the general selection are ordered by a
composite score. A border zone extends 10\% of capacity on either side of the
cutoff, and applicants in that zone who are judged to fit a school's admission
policy particularly well may be admitted ahead of their composite-score rank.
Thus, the admission-policy assessment is applied to an academically defined
eligibility band rather than as a second ranking of the entire applicant pool.
However, beginning with the 2028 admission cycle, schools may instead use a
school-specific track in the first selection stage for up to 50\% of capacity.
Only students who opt into that track are considered under its school-specific
materials and method. The remaining seats are then filled by the common composite
score.\footnote{See the official materials on the admissions reform issued by the
\href{https://www.pref.osaka.lg.jp/o180040/kotogakko/gakuji-g3/senbatsukaizen.html}
{Osaka Prefectural Board of Education}.}

This reform highlights a general design question addressed by our analysis:
how do a criterion's nominal seat share and its position in the sequence
translate into substantive influence over the composition of the admitted
class? Our theorems do not provide an overall evaluation for Osaka's reform because the changes span multiple areas. However, they provide a ceteris paribus comparative static that isolates
precedence: holding other conditions fixed,
moving a mission-specific block earlier cannot improve the admitted set
from that criterion's perspective and can make it strictly worse. In this way, our findings formally demonstrate that the seemingly self-evident policymaker intuition that selection based on more important criteria should be conducted earlier can be misleading, thereby contributing to more effective policymaking in real-world institutional design.

\section{Related Literature}

Our paper relates first to the literature that models institutions through
choice functions. Classical many-to-one matching assumes responsive
institutional preferences, while subsequent work identifies substitutability
and related conditions under which stable matchings exist and deferred
acceptance retains desirable properties
\citep{gale1962college,kelso1982job,roth1984stability,blair1988lattice,
hatfield2005matching,hatfield2010substitutes}. Our single-college choice rule is 
lexicographic: an ordered list of rankings is used to select
successively the highest-ranked remaining applicant
\citep{chambers2018lexicographic}. The same primitive underlies matching with
slot-specific priorities, where each slot has its own priority ranking and
slots are processed according to a precedence order
\citep{kominers2016matching}. Related work characterizes lexicographic choice
under variable capacity constraints \citep{dogan2021lexicographic} and extends
slot-specific-priority systems through dynamic reserves and capacity transfers
\citep{aygun2020dynamic,avataneo2021slot}. Whereas these papers primarily study
the properties of the resulting choice rules and matching mechanisms, we hold
fixed the component rankings and derive comparative statics with respect to
their precedence order.

A second strand studies institutional objectives over the composition of
admitted students. Matching with distributional constraints has been used to
analyze diversity requirements, minimum quotas, and regional or system-wide
caps
\citep{westkamp2013analysis,kamada2015efficient,kamada2024fair,
fragiadakis2017improving}. In school choice, reserve-based affirmative action
can avoid some adverse effects of majority quotas
\citep{kojima2012school,hafalir2013effective}, while soft bounds and
appropriately designed choice rules can reconcile diversity objectives with
stability and non-wastefulness
\citep{ehlers2014school,echenique2015control}. Further work studies
multidimensional privileges and overlapping vertical and horizontal
reservations
\citep{aygun2021multidimensional,sonmez2022affirmative}.
\citet{imamura2025meritocracy} formalizes the trade-off between meritocracy
and diversity in admissions and hiring. Our analysis is complementary: we
take the criterion-specific rankings and the number of seats assigned to each
criterion as given and ask how sequencing alone changes the admitted students.

Closest to our comparative-statics question is the literature on processing
order in reserve systems. \citet{dur2018reserve} show that the order in which
open and walk-zone-reserved seats are processed is a substantive design
choice and that reserve-first processing can have unintended effects.
\citet{dur2020explicit} study Chicago's four socioeconomic tiers and show
that seat-processing order is an additional instrument for explicit and
statistical targeting. \citet{pathak2023fair} formulate reserve systems with
multiple categories and category-specific priorities, define sequential
reserve matching through an order of precedence, and derive comparative
statics for category-specific maximum cutoffs. In a model allowing
overlapping types and priorities that differ across reserved copies,
\citet{almeer2024increasing} show how precedence can be modified to increase
the representation of a targeted type. Related order-sensitive designs arise
in immigration-lottery allocation \citep{pathak2025immigration}, and
behavioral evidence documents substantial misunderstanding of reserve
processing rules \citep{pathak2023reversing}. Rather than selecting a
particular sequential order, \citet{delacretaz2021processing} studies a
simultaneous reserve rule that does not rely on a precedence order. \citet{papai2022targeted} develop targeted priority reserves and
the broader class of DA mechanisms with sequential priority reserves,
distinguish representation from effective preferential treatment, and
compare alternative processing sequences. Most directly related to our DA
analysis, a favorable comparison of a college's choice sets need
not extend to a Pareto comparison of the DA mechanisms for priority
agents because different selection rules may generate different rejection
chains. 

Our model differs from reserve models in both its primitives and its welfare
comparison. In much of the reserve literature, category priorities are
generated from group eligibility and a baseline ranking. More recent models
allow multiple categories, overlapping eligibility, and category-specific
priorities, but our generalization is in a different direction. Criteria need
not represent eligibility groups, every criterion may rank all applicants,
and these rankings may be arbitrarily nonaligned. Moreover, each criterion
evaluates the entire admitted set according to its responsive preference,
rather than through the number of admitted members of a targeted type,
category-specific cutoffs, or the welfare of applicants.

Finally, our matching-market analysis relates to the literature on rejection
chains and large random matching markets
\citep{kojima2009incentives,ashlagi2017unbalanced}. We adapt these ideas to
show that feedback into the focal college becomes negligible under
regularity, thereby recovering the single-college comparison asymptotically.
The Boston mechanism and its strategic properties are studied by
\citet{abdulkadiroglu2003school} and \citet{ergin2006games}. Our Boston result,
which is pointwise in submitted rank-order lists and assumes
criterion-independent acceptability, highlights a different route to the
same comparison: final acceptance prevents subsequent feedback from revising
the focal college's class. Relatedly, \citet{chaudhury2024affirmative} study
a hybrid affirmative-action mechanism that uses immediate acceptance for
reserved seats and deferred acceptance otherwise. Their focus is on welfare,
fairness, and incentives, whereas ours provides a pointwise comparison of
the focal college's admitted sets across precedence orders.
Taken together, our contribution is to characterize the criterion-wise effects of precedence order under arbitrary rankings and to identify how these effects interact with applicant-set feedback and finality across mechanisms.

\section{Model}

There are finite sets $I$ and $C$ of students and colleges. Each student $i \in I$ has a strict preference order $\succ_i$ over the set of colleges $C \cup \{\emptyset\}$. Each college $c \in C$ has capacity $q_c \in \mathbb{Z}_{++}$. For convenience, the $q_c$ slots of each college $c$ are treated as distinct matching objects. Formally, each slot is represented by a pair $s = (c,k)$, where $c$ is the college and  
$k \in \{1,\dots,q_c\}$ is the slot index within that college. Let $S$ denote the set of all slots.

We consider a two-sided matching problem between students $I$ and slots $S$. We apply the student-proposing deferred acceptance (DA) algorithm and obtain the student-optimal stable matching. At the college level, each student $i$ is indifferent among all slots of the same college. To implement the sequential admission, we use a fixed refinement of the college-level preference: slots from different colleges are compared according to $\succ_i$, and slots within the same college are ranked by increasing slot index. Formally, we define a preference order $\succ_i'$ over slots as follows:
\begin{equation}
(c,k) \succ_i' (c',k')
\quad \text{if and only if either} \quad
\text{(i) } c \succ_i c'
\;\; \text{or} \;\;
\text{(ii) } c = c' \text{ and } k < k'.
\end{equation}
In addition, $(c, k)\succ_i' \emptyset$ if and only if $c \succ_i \emptyset$.
Let $\succeq_i'$ denote the weak preference relation induced by
$\succ_i'$.
This within-college refinement is part of the mechanism and is held fixed throughout the analysis. With slot-specific priorities, a different ordering of a college's slots can change its college-level admitted set; the increasing-index convention is what implements the precedence order studied below.

Each college $c$ has a finite set $T_c$ of admission criteria. Each criterion $t \in T_c$ is represented by a linear preference order $\succ_t$ over students $I \cup \{\emptyset\}$. We suppress the college index and write $\succ_t$ when focusing on one college. Among the $q_c$ slots of college $c$, some slots are assigned the preference order $\succ_t$ of criterion $t$, and these slot preferences are used on the slot side in DA. Which criterion is assigned to each slot is specified by a \emph{precedence order},
\begin{equation}
    v = (v(1), \dots, v(q_c)) \in T_c^{q_c},
\end{equation}
viewed as a finite sequence of criterion labels. If college $c$'s precedence order is $v$ and $v(k) = t$, then the preference of slot $(c,k)$ is $\succ_t$.

We use the following terminology for precedence orders. For
$w=(w(1),\ldots,w(q_c))$ and $h\in\{0,\ldots,q_c\}$, the sequence
$(w(1),\ldots,w(h))$ is the \emph{prefix of $w$ of length $h$}, and
$(w(h+1),\ldots,w(q_c))$ is the \emph{suffix of $w$ after the first
$h$ slots}. Two precedence orders $w$ and $w'$ share a
\emph{common prefix} $u\in T_c^h$ if
\[
    w(k)=w'(k)=u(k)
    \qquad\text{for every }k=1,\ldots,h.
\]
A \emph{$t$-block} is a consecutive segment of slots assigned to
criterion $t$; it need not be maximal. The prefix is empty when $h=0$,
the suffix is empty when $h=q_c$, and a displayed block is empty when
its stated length is zero.

For a criterion $t$, a set $Z\subset I$, and $r\in\mathbb Z_+$, let $ Z_t^+\coloneqq\{i\in Z:i\succ_t\emptyset\}$ and $\best_t(Z,r)$ denote the set of the $\min\{r,|Z_t^+|\}$ highest-ranked students in $Z_t^+$; in particular, $\best_t(Z,0)=\emptyset$. Let $\topacc_t(Z)$ denote the highest-ranked member of $Z_t^+$ when this set is nonempty, and set $\topacc_t(Z)=\emptyset$ otherwise. Given a precedence order $v$ and an applicant set $X$, a college's choice $\Ch_v(X)$ is defined recursively as follows: at slot $k$, criterion $v(k)$ selects the best applicant with respect to $v(k)$ from those who have not been selected at earlier slots. If this value is $\emptyset$, the slot is left vacant. Thus, $r$ consecutive $t$-slots facing $Z$ select exactly $\best_t(Z,r)$, whose cardinality can be smaller than $r$.

We focus on a focal college $c_0$. 
We fix all students' preferences and all criterion rankings, as well as
the precedence orders of all colleges other than $c_0$.
We then analyze how the set of students matched with college $c_0$ changes when college $c_0$ varies its precedence order $v$.

For each criterion $t$, fix an arbitrary complete and transitive preference relation $\succeq_t$ over subsets of students that is responsive to $\succ_t$. Formally, for every set $S$ and distinct students $i_1,i_2\in I\setminus S$, (i) $S\cup\{i_1\}\succ_t S\cup\{i_2\}$ if and only if $i_1\succ_t i_2$, and (ii) $S\cup\{i_1\}\succ_t S$ if and only if $i_1\succ_t\emptyset$. Thus, deleting a student whom $t$ ranks below $\emptyset$ weakly improves a set. A responsive extension need not be unique. Every set comparison below follows from a finite sequence of improving replacements, additions of acceptable students, or deletions of unacceptable students, and therefore holds for every such extension.

For our analysis under student-proposing DA, we impose no
restrictions on students' preferences, nor on the preferences or precedence
orders of colleges other than $c_0$. Therefore, for the DA sections, without loss
of generality, we may represent every college other than $c_0$ as a collection
of capacity-one colleges whose objects are ranked consecutively by every student.
For example, suppose that a college $c^*$ has two seats and two criteria $A$ and $B$,
and that its precedence order assigns these criteria in this order.
For the purpose of analyzing college $c_0$, this situation is equivalent to having two colleges $c_A$ and $c_B$, each with capacity one and a single criterion, where the preference of $c_A$ is the same as the preference of criterion $A$, the preference of $c_B$ is the same as the preference of criterion $B$, and in every student's preference list, $c_B$ is listed immediately after $c_A$ whenever $c_A$ or $c_B$ appears. Slightly abusing notation, we drop the subscript indicating the college from $c_0$'s capacity $q_{c_0}$ and the set of criteria $T_{c_0}$ and simply write them as $q$ and $T$.

We study how the set of students matched with college $c_0$ responds when college $c_0$ changes its precedence order. In particular, we examine whether, from the perspective of each criterion of college $c_0$, the resulting set of matched students becomes more favorable. We fix the preferences of all students and of all colleges other than $c_0$. When college $c_0$ adopts a precedence order $v$, let $\bmu(v)$ denote the matching produced by the student-proposing DA, and let $\mu(v)$ denote the set of students matched with college $c_0$ under this matching. In particular, the student matched to the $k$-th seat of college $c_0$ is denoted by $\mu_k(v)$.

\section{Single College's Choice Problem}

We begin by considering a single college's choice problem, rather than the full matching market. Maintaining the assumption that the matching is determined by the student-proposing DA, this case is equivalent to a market in which $c_0$ is the only college, that is, $C = \{c_0\}$. Without loss of generality, within this section, we assume that every student $i$ accepts college $c_0$: $c_0 \succ_i \emptyset$. In this environment, college $c_0$'s choice solely determines its matched students: $\mu(v) = \Ch_v(I)$.

In the preference lists displayed below, $\emptyset$ is understood to
rank below all listed alternatives when it is omitted, and any unlisted
alternative is ranked below $\emptyset$.

\subsection{Expansion}

We first consider the case in which the number of slots assigned to criterion $t$ is increased. Since we focus on a single college's choice problem, one might expect that increasing the number of slots evaluated by criterion $t$ would make the set of students admitted by college $c_0$ more favorable with respect to $\succ_t$. The following example shows, however, that this intuition does not always hold.

\begin{example}\label{exa:expansion_three}
Let $I=\{i_1,i_2,i_3\}$, $q=2$, and $T=\{X,Y,Z\}$. Suppose the preference orders of the criteria are as follows:
\begin{align}
    \succ_{X}:&\: i_1, i_2, i_3;\\
    \succ_{Y}:&\: i_2, i_1, i_3;\\
    \succ_{Z}:&\: i_1, i_3, i_2.
\end{align}
We can verify that $\Ch_{Y, Z}(I) = \{i_1, i_2\}\succ_{X}\Ch_{X, Z}(I) = \{i_1, i_3\}$, even though $(X, Z)$ is obtained from $(Y, Z)$ by replacing the slot assigned to criterion $Y$ with a slot assigned to criterion $X$.
\end{example}

In Example~\ref{exa:expansion_three}, although the number of slots assigned to criterion $X$ is increased, the set of students admitted by college $c_0$ becomes worse with respect to $\succ_X$. The reason is that, by not admitting student $i_1$, who is most preferred by criterion $X$,
at the first slot, criterion $X$ can induce criterion $Z$ to admit student $i_1$ at the second slot. When the first slot is assigned to criterion $Y$, this outcome arises naturally. Moreover, the student most preferred by criterion $Y$ is $i_2$, who is the second-most preferred
student of criterion $X$. As a result, when the precedence order is $(Y, Z)$, the admitted set is $\Ch_{Y, Z}(I) = \{i_1,i_2\}$, which is ideal from the perspective of criterion $X$. In contrast, when the criterion assigned to the first slot is replaced, and the precedence order becomes $(X, Z)$, the outcome becomes less favorable for criterion $X$: $\Ch_{X, Z}(I) = \{i_1, i_3\}$.

The structure of this problem resembles the incentives faced by receivers to misreport their preferences in DA. In DA, for example, a receiver with true preference $\succ_X$ does not benefit, in a direct
sense, from misreporting its preference as $\succ_Y$ in terms of improving the students it keeps at that stage. However, by changing which students are rejected, such misreports can alter the entire market situation and may ultimately be beneficial. In our multi-criterion admissions setting, the student admitted at a given slot affects which students remain available for downstream slots. Through this channel, it may fail to be optimal for each criterion to myopically admit the student it ranks highest at its own slot.

This issue arises whenever there exists a slot assigned to another criterion that comes after the slot newly assigned to a given criterion. In contrast, if no such slot exists, no strategic incentive arises, and each criterion optimally admits its most preferred student. Theorem~\ref{thm:category_slot_increasing} formalizes this observation.

\begin{thm}\label{thm:category_slot_increasing}
Take any criterion $t\in T$ and any precedence order $v$. Choose $k^* \in \{1,\dots,q\}$ arbitrarily, and define a precedence order $v'$ by $v'(k^*) \coloneqq t$ and $v'(k) \coloneqq v(k)$ for all $k \neq k^*$. If $v(k) = t$ for all $k > k^*$, then, for any $X \subset I$, $\Ch_{v'}(X) \succeq_t \Ch_{v}(X)$.
\end{thm}

The full proofs are presented in Appendix~\ref{sec:proofs}.

By repeatedly applying Theorem~\ref{thm:category_slot_increasing}, we can also show that an expansion of criterion $t$ that replaces all slots in the last segment of the precedence order with criterion $t$ always makes the set of admitted students more favorable with respect to $\succ_t$.

\begin{cor}\label{cor:repeatedly_apply}
    In the single-college choice problem, take any precedence order $v$ and criterion $t\in T$, and define a precedence order $v^m$ by $v^m(k) \coloneqq v(k)$ for $k = 1, \dots, q - m$ and $v^m(k) \coloneqq t$ for $k = q - m + 1, \dots, q$. Then, for every $m \in \{0, \dots, q\}$ and for every $X \subset I$, $\Ch_{v^m}(X) \succeq_t \Ch_v(X)$.
\end{cor}

\subsection{Swap}

Next, we examine the effects of changing the order of criteria in the precedence order while keeping the number of slots assigned to each criterion fixed. \citet{dur2018reserve} show that, when there are two types of students, reserve-eligible students and open-category (reserve-ineligible) students, placing the open category before the reserve-eligible
category makes every reserve-eligible student more likely to be admitted than the reverse order. This model is a special case of ours in which the reserve-eligible category strictly prefers reserve-eligible students to open-category students. This result suggests that the criterion positioned later in the precedence order
may be better respected. However, in our model, criteria may have arbitrary preferences, and a college may have three or more criteria.
Therefore, we need to examine whether the result of \citet{dur2018reserve} continues to hold in this more general setting.

The following example demonstrates that, when two criteria are each assigned one slot, the admitted set is weakly more favorable to the criterion processed later.

\begin{example}\label{exa:single_college_swap_two_categories}
Let $I = \{i_1, i_2, i_3, \dots, i_n\}$ for some $n \ge 3$, $q = 2$, and $T = \{A, B\}$. Suppose that every student is acceptable to both criteria. Without loss of generality, we assume that criterion $A$'s preference is given by $\succ_A: i_1, i_2, i_3, \dots, i_n$. We consider all possible cases for $\succ_B$ and compare exhaustively which precedence order,
$(A,B)$ or $(B,A)$, is more favorable with respect to $\succ_A$.

If the student most preferred by criterion $B$ is different from the one most preferred by criterion $A$ (i.e., $i_1$), then under both precedence orders $(A,B)$ and $(B,A)$, each criterion admits its most preferred student, and the admitted set does not depend on the precedence order. If the student most preferred by criterion $B$ is $i_1$, and the second-most preferred student of criterion $B$ is also the same as that of criterion $A$, namely $i_2$, then the admitted set is $\{i_1,i_2\}$ regardless of the precedence order.

However, when the student most preferred by criterion $B$ is $i_1$ but its second-most preferred student is not $i_2$, the outcome depends on the precedence order. Student $i_1$ is ranked highest by both criteria $A$ and $B$ and is therefore admitted at the first slot regardless of the precedence order. The student ranked second, however, differs across criteria. When criterion $A$ is assigned the second slot, it can admit student $i_2$, who is its second-most preferred student. In contrast, when criterion $A$ is assigned the first slot, the second slot is filled by the student most preferred by criterion $B$ among the remaining students, who is not necessarily preferred by criterion $A$. As a result, in this case, the precedence order $(B,A)$, in which criterion $A$ is positioned later, yields an admitted set that is more favorable for criterion $A$
than the precedence order $(A,B)$.
\end{example}

In many admissions settings, multiple evaluation criteria agree on the very best applicants, while they rank mid-ranked applicants differently. As Example~\ref{exa:single_college_swap_two_categories} illustrates, such disagreements at the margin can make the relative position of criteria in the precedence order consequential for the composition of the admitted set of students.

This finding can be generalized. First, since we consider a single college's choice problem, the presence of slots or criteria that appear earlier in the precedence order than the slots being swapped does not affect the argument. Second, when two criteria each have multiple slots, the analysis becomes more involved, as it requires a classification of cases in which the sets of students favored by the two
criteria coincide or differ. Nevertheless, it can be shown that
the same qualitative result continues to hold. The following theorem formalizes these observations.

\begin{thm}\label{thm:category_utility_increasing}
Fix two criteria $A,B\in T$ and integers
$\bar q,q_A,q_B\in\mathbb Z_+$ such that
$\bar q+q_A+q_B=q$. Fix a sequence $v\in T^{\bar q}$, and define
the precedence orders
\begin{align}
v_{AB}
&\coloneqq
\bigl(
v,
\underbrace{A,\ldots,A}_{q_A},
\underbrace{B,\ldots,B}_{q_B}
\bigr),
\label{eq:vAB}\\
v_{BA}
&\coloneqq
\bigl(
v,
\underbrace{B,\ldots,B}_{q_B},
\underbrace{A,\ldots,A}_{q_A}
\bigr).
\label{eq:vBA}
\end{align}
Thus, the two orders share the prefix $v$ and differ only by reversing
their final $A$- and $B$-blocks. Then, for every $X\subset I$,
\begin{equation}
\Ch_{v_{BA}}(X)\succeq_A\Ch_{v_{AB}}(X)
\quad\text{and}\quad
\Ch_{v_{AB}}(X)\succeq_B\Ch_{v_{BA}}(X).
\label{eq:two_block_comparison}
\end{equation}
\end{thm}

In the proof, to show that $v_{AB}$ is more desirable for criterion $B$ than $v_{BA}$, we repeatedly apply two operations, \emph{swap} and \emph{relabel}, to $v_{AB}$. In the single-capacity environment illustrated in Example~\ref{exa:single_college_swap_two_categories}, a swap corresponds to changing the precedence order from $(A, B)$ to $(B, A)$, whereas a relabel corresponds to changing the precedence order from $(A, B)$ to $(B, B)$. These operations are chosen so that the admitted set remains unchanged, while the precedence order is gradually transformed to approach $v_{BA}$.

For notational simplicity, we ignore the first $\bar{q}$ slots, which are occupied by an arbitrary but common precedence order $v$, and count the $(\bar{q}+1)$-th slot as the first slot. At each step $k=1,\ldots,q_B$, unless the procedure has already stopped, we apply either swap or relabel at the $k$-th slot
and replace the criterion $A$ placed at that slot with $B$.
Accordingly, at step $k$, the resulting precedence order can be written as
\begin{equation}\label{eq:wkm1_form}
w = \bigl(v,\underbrace{B,\ldots,B}_{k-1},\underbrace{A,\ldots,A}_{a'},\underbrace{B,\ldots,B}_{b'}\bigr).
\end{equation}
The values of $b'$, the number of $B$ slots in the final segment, and $a'$, the number of $A$ slots in the preceding segment, depend on how many times swap and relabel have been applied up to that point.

A swap is an operation that replaces the criterion $A$ placed at the $k$-th slot with criterion $B$, and at the same time replaces the criterion $B$ placed at the $(k+a')$-th slot with criterion $A$. This operation increases the number of criterion $B$ slots in the leading segment, while decreasing the number of criterion $B$ slots in the final segment. As a result, the total number of slots assigned to criteria $A$ and $B$ remains unchanged.

In contrast, a relabel simply replaces the criterion $A$ placed at the $k$-th slot with criterion $B$. This operation increases the number of criterion $B$ slots in the leading segment, and does not affect the final segment. Therefore, each application of relabel increases the total number of slots assigned to criterion $B$.

After applying the swap or relabel to $w$ (specified by \eqref{eq:wkm1_form}), the precedence order becomes the following:
\begin{align}
\swap(w,k)
&\coloneqq \bigl(v,\underbrace{B,\ldots,B}_{k},\underbrace{A,\ldots,A}_{a'},\underbrace{B,\ldots,B}_{b'-1}\bigr),\label{eq:wswap_def}\\
\relab(w, k)
&\coloneqq \bigl(v,\underbrace{B,\ldots,B}_{k},\underbrace{A,\ldots,A}_{a'-1},\underbrace{B,\ldots,B}_{b'}\bigr).\label{eq:wrelab_def}
\end{align}

In the proof, we propose the boundary invariance condition (Condition~\ref{condition:boundary_invariance}) as the necessary and sufficient condition for a swap not to change the admitted set. 
In Example~\ref{exa:single_college_swap_two_categories}, we showed that
the outcome is unchanged (i.e., $\Ch_{A, B}(X) = \Ch_{B, A}(X)$) in either of two cases: (i) the two criteria
have different top-ranked applicants, or (ii) they have the same
top-ranked applicant and the same second-ranked applicant.
The boundary invariance condition extends this insight to cases with larger capacities.

Conversely, we show that when the boundary invariance condition is not satisfied for precedence
order $w$ at the $k$-th slot, applying relabel does not change the admitted set: $\Ch_w(X) = \Ch_{\relab(w,k)}(X)$.
This corresponds to the case in Example~\ref{exa:single_college_swap_two_categories} in which (i) criteria $A$ and $B$ share the same most preferred student, while (ii) their second-most preferred students differ. In this situation, in the single-capacity case, the same student is admitted at the first slot regardless of whether the criterion is $A$ or $B$: $\Ch_{A, B}(X) = \Ch_{B, B}(X)$. An analogous argument applies in the case of larger capacities. When the boundary invariance condition is not satisfied, simply changing the label of the $k$-th slot from $A$ to $B$ does not change the admitted set. Therefore, in this case, we apply relabel rather than swap.

Starting from $v_{AB}$ and repeating this operation at most $q_B$ times, we obtain the following precedence order as the final outcome:
\begin{equation}\label{eq:vstar_form}
v^* = \bigl(v,\underbrace{B,\ldots,B}_{q_B},\underbrace{A,\ldots,A}_{q_A^*},\underbrace{B,\ldots,B}_{q_A-q_A^*}\bigr).
\end{equation}
By construction, we have $\Ch_{v^*}(X) = \Ch_{v_{AB}}(X)$. Moreover, $v_{BA}$ can be obtained from $v^*$ by replacing the $B$'s in the last $(q_A-q_A^*)$ slots with $A$. Therefore, by Corollary~\ref{cor:repeatedly_apply}, it follows that $\Ch_{v_{AB}}(X) = \Ch_{v^*}(X) \succeq_B \Ch_{v_{BA}}(X)$.

Explicitly constructing a precedence order $v^*$ that yields the same outcome as $v_{AB}$ is useful for quantifying the effect of changing the precedence order. Both $v^*$ and $q_A^*$ depend on the fixed applicant set $X$; we suppress this dependence to simplify notation. The difference between $v^*$ and $v_{BA}$ lies in the criteria assigned to the last $(q_A - q_A^*)$ slots, and this difference grows with each application of relabel. We apply swap whenever it leaves the admitted set unchanged, and use relabel only when swap would change the outcome. Thus, modifying the precedence order has an effect comparable to reallocating $q_A - q_A^*$ slots.

The magnitude of this effect depends on how criteria $A$ and $B$ rank the applicants. To illustrate, suppose that after the common prefix is processed, the remaining applicants consist of a common top group of size $m=q_A=q_B$ and a residual group of more than $m$ applicants. Suppose that both criteria find all of them acceptable, rank the common top group identically above the residual group, and rank the residual group in exactly opposite orders. The terminal $B$-block then does not select the $A$-best residual applicant at any step, so every operation is a relabeling: $q_A^*=0$ and $v^*=(v,B,\ldots,B)$. The residual-size qualification is essential. If at most $m$ applicants remain below the common top group, the terminal $B$-block can select all of them, swaps can occur, and $q_A^*$ need not be zero. In contrast, if the rankings are identical, or if the sets of applicants acceptable to $A$ and $B$ are disjoint, every swap is outcome-invariant. As a result, $q_A^*=q_A$ and $v^*=v_{BA}$. In general, it is difficult to predict the value of $q_A^*$. However, given the criteria and applicant set, the logic developed in our proof can be used to compute it efficiently.

Theorem~\ref{thm:category_utility_increasing} shows that, when a swap is applied to the last segment of the precedence order, the resulting admitted set becomes more favorable for the later-positioned criterion and less favorable for the earlier-positioned criterion. The following example demonstrates that this conclusion does not necessarily hold when the swap is not restricted to the last segment.

\begin{example}\label{exa:swap_single_college_three_categories}
Let $I = \{i_1, i_2, i_3, i_4\}$, $q = 3$, and $T = \{X, Y, Z\}$. The preferences of criteria $X, Y, Z$ are as follows:
\begin{align}
    \succ_X: & \: i_1, i_2, i_3, i_4;\\
    \succ_Y: & \: i_1, i_3, i_2, i_4;\\
    \succ_Z: & \: i_2, i_4, i_1, i_3.
\end{align}
In this case, $\Ch_{X, Y, Z}(I) = \{i_1, i_2, i_3\} \succ_X \Ch_{Y, X, Z}(I) = \{i_1, i_2, i_4\}$, implying that the admitted set becomes more favorable with respect to $\succ_X$ when criterion $X$ is positioned earlier.
\end{example}

Note that Example~\ref{exa:swap_single_college_three_categories} can be interpreted as a case of reserves for multiple targeted groups. Criterion $X$ can be interpreted as the open category, criterion $Y$ as a reserve for the targeted group $\{i_1,i_3\}$, and criterion $Z$ as a reserve for the targeted group $\{i_2,i_4\}$. Criteria $Y$ and $Z$ are assumed to (i) prioritize students in their respective targeted groups over non-targeted students, and (ii) rank students within the same group according to a common preference ranking $\succ_X$. Therefore, the insight that, with three or more criteria, a swap outside the last segment can either benefit or harm the later-positioned criterion continues to hold even under structures motivated by affirmative action or reserve policies for targeted groups.

Criteria $X$ and $Y$ both rank $i_1$ as their most preferred student. When the precedence order is $(Y,X,Z)$, criterion $X$ succeeds in having criterion $Y$ admit $i_1$ without using its own slot, and then admits its second-most preferred student $i_2$. However, the subsequent criterion $Z$ admits student $i_4$, who is the least preferred student from the perspective of criterion $X$.

In contrast, when the precedence order is $(X,Y,Z)$, criterion $X$ uses its own slot to admit $i_1$, and criterion $Y$ admits $i_3$, who is its second-most preferred student and the third-most preferred student of criterion $X$. At this stage, criterion $X$ is worse off by moving to an earlier position. Nevertheless, student $i_2$, who is displaced by this swap, is subsequently admitted by criterion $Z$. As a result, all three students most preferred by criterion $X$, namely $\{i_1,i_2,i_3\}$, are ultimately admitted.

The intuition of Example~\ref{exa:swap_single_college_three_categories} is similar to that of Example~\ref{exa:expansion_three}. When a swap is applied at an upstream segment in the precedence order, the criterion that is positioned later is better respected for the students admitted at those upstream slots. However, such a swap also changes which students are candidates for downstream slots. This, in turn, affects the set of students admitted downstream. Consequently, the overall effect of an upstream swap on the final admitted set cannot be determined in a uniform manner.

\section{Deferred Acceptance in Finite Markets}

In the single-college choice problem, we have shown that, when the order of two criteria in the final segment is changed, the criterion positioned later is advantaged. However, this result does not necessarily extend to a matching market with multiple colleges. The following example illustrates this point.

\begin{example}\label{exa:matching_market_two_categories}
Let $I = \{i_1, i_2, i_3\}$, $q = 2$, $T = \{A, B\}$, and $C = \{c_0, c_1\}$. The preference orders of students $i_1, i_2, i_3$, criteria $A$ and $B$, and college $c_1$ are as follows:
\begin{equation}
    \begin{array}{@{}l @{\qquad\qquad} l@{}}
    \begin{aligned}
        \succ_{i_1}: & \; c_0\\
        \succ_{i_2}: & \; c_1, c_0\\
        \succ_{i_3}: & \; c_0, c_1
    \end{aligned}
    &
    \begin{aligned}
        \succ_A:     & \; i_1\\
        \succ_B:     & \; i_2, i_1, i_3\\
        \succ_{c_1}: & \; i_3, i_2
    \end{aligned}
    \end{array}
\end{equation}
We can verify that $\mu(B, A) = \{i_1, i_2\} \succ_B \mu(A, B) = \{i_1, i_3\}$. Note that under $v = (B, A)$, college $c_1$ matches with student $i_3$, whereas under $v = (A, B)$, college $c_1$ matches with student $i_2$.
\end{example}

In Example~\ref{exa:matching_market_two_categories}, the admitted set of college $c_0$ is more favorable with respect to $\succ_B$ when the precedence order is $(B,A)$ than when it is $(A,B)$. This contrasts with the single-college setting, in which the criterion positioned later is always advantaged.

The underlying logic is similar to the reason why a college on the receiver side of the student-proposing DA may benefit from misreporting its preferences. When the precedence order is $(A,B)$, the set $\{i_1,i_3\}$ that applies in the first step is kept by college $c_0$ and becomes its admitted set. In contrast, when the precedence order is $(B,A)$, criterion $B$ keeps $i_1$ at the first slot, and then criterion $A$ has no student it wishes to keep at the second slot. Consequently, student $i_3$ is rejected by college $c_0$. At this point, the set of students kept by college $c_0$ is $\{i_1\}$, which is worse than $\{i_1,i_3\}$ obtained under the precedence order $(A,B)$, consistent with the single-college case.

In a matching market, however, a rejected student may apply to another college and trigger a rejection chain. Here, student $i_3$, who is rejected by college $c_0$, applies to college $c_1$ and is kept tentatively. This causes student $i_2$, who was previously matched with $c_1$, to be rejected. Student $i_2$ then applies to her second-choice college $c_0$.
Since criterion $B$ can keep $i_2$, the final admitted set becomes $\{i_1,i_2\}$, which is an improvement over $\{i_1,i_3\}$ obtained under $(A,B)$.

Thus, in a matching market, the precedence order affects which students are rejected, and these rejections can trigger rejection chains. Since myopically keeping the highest-ranked students need not lead to the best final matching outcome, a characterization such as Theorem~\ref{thm:category_utility_increasing} does not hold in the matching market.

\section{Large Markets}\label{sec:large_market}

Example~\ref{exa:matching_market_two_categories} identifies the only possible channel through which the
single-college comparison can be reversed in a matching market.  A change in the precedence
order changes which students are rejected by the focal college.  A rejected student may then
start a rejection chain, and the chain may eventually induce a student who did not previously
apply to the focal college to apply there.  This section first
shows that, in the absence of such feedback, the matching-market comparison is exactly the same as
the single-college comparison.  We then show that such feedback vanishes in a large market model.

For a precedence order $w$ of college $c_0$, let
\begin{equation}
    \mathcal A(w)
    \coloneqq
    \{i\in I:\text{$i$ proposes to a slot of $c_0$ during DA under $w$}\}
\end{equation}
be the set of students who ever apply to $c_0$.  Because every student ranks the slots of a given
college consecutively and in increasing slot order, a student belongs to $\mathcal A(w)$ if and only
if she proposes to the first slot of $c_0$. The following is immediate.

\begin{lem}\label{lem:proposal_set_representation}
For every precedence order $w$,
\begin{equation}
    \mu(w)=\Ch_{w}(\mathcal A(w)).
\end{equation}
Moreover, the student assigned to each slot of $c_0$ under $\bmu(w)$ is the student selected by
that slot in the sequential construction of $\Ch_{w}(\mathcal A(w))$.
\end{lem}

Lemma~\ref{lem:proposal_set_representation} implies that, when the set of applicants is fixed, the admitted set of college $c_0$ is determined by the single-college choice problem. In such a case, the single-college comparison of Theorem~\ref{thm:category_utility_increasing} can be applied.

\begin{prop}\label{prop:common_applicant_set}
If $\mathcal A(v_{AB})=\mathcal A(v_{BA})$, then
\begin{equation}
    \mu(v_{BA})\succeq_A\mu(v_{AB})
    \qquad\text{and}\qquad
    \mu(v_{AB})\succeq_B\mu(v_{BA}).
\end{equation}
\end{prop}

We next define an auxiliary transition process. Fix an arbitrary strict order
$\triangleright$ on $I$. Consider two precedence orders $w,w'$ of
$c_0$, holding all other primitives fixed, and put
\begin{equation}
\begin{alignedat}{2}
    X&\coloneqq\mathcal A(w),\qquad
    &D&\coloneqq\Ch_{w'}(X),\\
    R(w,w')&\coloneqq\mu(w)\setminus D,\qquad
    &N(w,w')&\coloneqq D\setminus\mu(w).
\end{alignedat}
\end{equation}
We construct the initial state of the process from $\bmu(w)$ as follows.
First, remove the students assigned to the slots of $c_0$ under $\bmu(w)$.
Next, apply the choice procedure associated with $w'$ to the applicant set
$X$. For each $k=1,\ldots,q$, if criterion $w'(k)$ selects a student at step
$k$ of this procedure, assign that student to slot $(c_0,k)$; otherwise leave
slot $(c_0,k)$ vacant. The set of students assigned to $c_0$ in this way is
$D=\Ch_{w'}(X)$. Thus, a student in $\mu(w)\cap D$ remains assigned to $c_0$
(possibly at a different slot), a student in $R(w,w')$ is released from
$c_0$, and a student in $N(w,w')$ is newly assigned to $c_0$.

If a student $i\in N(w,w')$ was assigned under $\bmu(w)$ to a slot $s$
outside $c_0$, remove $i$ from $s$ and place at $s$ an inactive placeholder,
denoted by $i^*$. The placeholder is not an agent and has no preference or assignment of its
own. Operationally, a student $j$ displaces $i^*$ at slot $s$ if and only if
$j\succ_s i$; otherwise, slot $s$ rejects $j$. Its sole purpose is to
preserve the priority cutoff that slot $s$ had under
$\bmu(w)$. If $i$ was unmatched under $\bmu(w)$, no placeholder is created.
Assignments at all other slots outside $c_0$ are initially unchanged. We call
an actual student assigned to a slot outside $c_0$ its \emph{real holder}, to
distinguish that student from a placeholder.

Take the students in $R(w,w')$ one at a time in
$\triangleright$-order. For each such student $i$, designate $i$ as the
\emph{active student}. Student $i$ considers, in descending order of her
refined preferences, all slots that she finds acceptable and ranks strictly
below every slot of $c_0$, and proposes to the highest-ranked such slot $s$.

Whenever an active student $j$ proposes to a slot $s$, if $\emptyset\succ_s j$, then $s$ rejects $j$. Suppose instead
that $j\succ_s\emptyset$. If $s$ is vacant, it holds $j$. If $s$ has a real
holder $h$, then it rejects $j$ when $h\succ_s j$; when $j\succ_s h$, it holds
$j$ and displaces $h$. If $s$ has a placeholder $\ell^*$, then it rejects $j$
when $\ell\succ_s j$; when $j\succ_s\ell$, it replaces $\ell^*$ with $j$.

After a rejection, the same active student proposes to her
next lower-ranked acceptable slot. If a real holder $h$ is displaced from a
slot $s$ outside $c_0$, then $h$ becomes the new active student and proposes,
in descending order of her refined preferences, slots below
$s$, beginning with the highest-ranked one. Call this slot $s'$ and repeat the same procedure as in the case of $s$ above. If a placeholder $\ell^*$ is
replaced, the chain ends. A student who is held at a vacant slot, or who replaces a
placeholder, makes no further proposal in the current chain. The chain also
ends if its active student has no lower-ranked acceptable slot remaining.

If the first slot of $c_0$ would be the next slot proposed
by an active student in $I\setminus X$, the process stops before that proposal
is made. Note that an
active student in $X$ always starts or resumes strictly below the entire block
of $c_0$ because of the following: A student in $R(w,w')$ starts below the entire block of
$c_0$. A student in $X\setminus(D\cup\mu(w))$ was rejected by every slot of
$c_0$ under the baseline DA and
can first become active only by being displaced from her
baseline assignment outside $c_0$; she then resumes immediately below that
assignment, which is itself below the entire block of $c_0$. Finally, a student in
$D$ is assigned to $c_0$ and does not become active when her placeholder is
displaced. Hence no student in $X$ can newly propose to the first slot of $c_0$ in
this process; only a student in $I\setminus X$ can do so.

The resulting sequence of proposals and displacements,
beginning with $i$ and continuing until one of the preceding stopping
conditions is met, is called the \emph{displacement chain initiated by $i$}.
After a chain ends, the process starts the chain of the next student in
$R(w,w')$, if any, from the assignment and placeholder configuration left by
the preceding chains.

The auxiliary \emph{$(w,w')$-transition process returns to $c_0$}
if, during a displacement chain, the first slot of $c_0$ is the next slot to
which an active student in $I\setminus X$ would propose under the procedure.
If every displacement chain ends without such a proposal becoming the next
proposal, the process does not return.

\begin{lem}\label{lem:no_return_inclusion}
If the $(w,w')$-transition process does not return to $c_0$, then
\begin{equation}
    \mathcal A(w')\subset \mathcal A(w).
\end{equation}
\end{lem}

Applying Lemma~\ref{lem:no_return_inclusion} in both directions gives the deterministic result
that will be used below.

\begin{thm}\label{thm:no_feedback_order_comparison}
Suppose that neither the $(v_{AB},v_{BA})$-transition process nor the
$(v_{BA},v_{AB})$-transition process returns to $c_0$.  Then
\begin{equation}
    \mu(v_{BA})\succeq_A\mu(v_{AB})
    \qquad\text{and}\qquad
    \mu(v_{AB})\succeq_B\mu(v_{BA}).
\end{equation}
\end{thm}

We now embed the preceding deterministic result in the random-market model of
\citet{kojima2009incentives}.  A random market is a tuple
\begin{equation}
    \widetilde\Gamma=
    \left(C,I,(q_c)_{c\in C},
    \bigl(T_c,(\succ_{c,t})_{t\in T_c}\bigr)_{c\in C},
    k,\mathcal D,\mathbf v\right),
\end{equation}
where $k$ is the length of every student's acceptable-college list,
$\mathcal D=(p_c)_{c\in C}$ is a distribution with $p_c>0$ for every $c$, and $\mathbf v$ is a
profile of precedence orders. Each student independently draws $k$ distinct colleges
sequentially from $\mathcal D$, as in \citet{kojima2009incentives}, and ranks the slots of each
drawn college consecutively in that college's precedence order.

A sequence $(\widetilde\Gamma^n)_{n\geq 1}$, with $|C^n|=n$, is \emph{regular} if there are
constants $k\in\mathbb Z_{++}$ and $\bar q\in\mathbb Z_{++}$ such that, for every $n$,
\begin{enumerate}
    \item every student lists exactly $k$ colleges;
    \item $q_c\leq\bar q$ for every $c\in C^n$;
    \item $|I^n|\leq\bar q n$; and
    \item every slot finds every student acceptable.
\end{enumerate}
The criteria's rankings may otherwise be arbitrary and may vary with $n$.

For the uniform result concerning a specified college, we use the sufficient-thickness condition
of \citet{kojima2009incentives}.  Let $p_{\max}^n=\max_{c\in C^n}p_c^n$, let
$N_c^n$ be the number of students whose complete $k$-college list contains $c$, and define
\begin{equation}
    V_T^n
    \coloneqq
    \left\{c\in C^n:\frac{p_{\max}^n}{p_c^n}\leq T
    \text{ and }N_c^n<q_c\right\},
    \qquad
    Y_T^n\coloneqq |V_T^n|.
\end{equation}
The sequence is \emph{sufficiently thick} if $\mathbb E[Y_T^n]\to\infty$ for some finite $T$.

For each college $c$, fix two local criteria $A_c,B_c\in T_c$ and two precedence orders
$v_{A_cB_c}^c$ and $v_{B_cA_c}^c$ that have a common prefix and differ only by reversing their
final $A_c$- and $B_c$-blocks, as in
Theorem~\ref{thm:category_utility_increasing}, and let $\mu_c^n(w)$ be the set of students matched with college $c$ when its precedence order is $w$.  All other colleges' precedence orders are held
fixed when the two outcomes for college $c$ are compared.  Let $\mathcal G_c^n$ be the event
\begin{equation}
    \mathcal G_c^n=\left\{
    \mu_c^n(v_{B_cA_c}^c)\succeq_{c,A_c}\mu_c^n(v_{A_cB_c}^c)
    \text{ and }
    \mu_c^n(v_{A_cB_c}^c)\succeq_{c,B_c}\mu_c^n(v_{B_cA_c}^c)
    \right\},
\end{equation}
where $\succeq_{c,t}$ is the responsive set preference associated with $\succ_{c,t}$.

\begin{thm}\label{thm:large_market_order_comparison}
Suppose $(\widetilde\Gamma^n)_{n\geq1}$ is regular.
\begin{enumerate}
    \item\label{part:1_large_market_order_comparison} The expected fraction of colleges at which at least one of the two inequalities fails converges to zero as $n\rightarrow\infty$:
    \begin{equation}
        \frac{1}{n}\sum_{c\in C^n}\Pr[(\mathcal G_c^n)^c]\longrightarrow 0.
    \end{equation}

    \item\label{part:2_large_market_order_comparison} If the sequence is also sufficiently thick, then the conclusion is uniform over
    colleges:
    \begin{equation}
        \sup_{c\in C^n}\Pr[(\mathcal G_c^n)^c]\longrightarrow 0.
    \end{equation}
    In particular, for any specified sequence of colleges $(c_0^n)_{n\geq1}$, the probability that either inequality fails converges to zero.
\end{enumerate}
\end{thm}

Theorem~\ref{thm:large_market_order_comparison} completes the
large-market argument. A later-positioned criterion may be disadvantaged only when a transition process returns to
the originating college in at least one direction and thereby alters the set
of students who apply there.  Regularity makes such return events negligible
for all but a vanishing fraction of colleges in expectation. Although regularity does
not by itself guarantee the comparison for an ex ante specified college, sufficient thickness strengthens the result by providing enough absorbing
colleges uniformly across the market, so that the return probability vanishes
for every college.  Once such feedback is absent, the two
precedence orders face the same applicant set, and the results established for the single-college choice problem apply.

\section{Boston Mechanism}\label{sec:boston}

All matching-market results up to this point have used student-proposing DA.
For that mechanism, treating a multi-slot college's slots that share the
same criterion as consecutively ranked unit-capacity objects preserves
the outcome.
The analogous transformation does not generally
preserve the outcome of the Boston mechanism: if the slots in one college were treated as separate colleges, a student rejected
by one internal slot would propose to the next internal slot only in a later round.  We therefore
define the Boston mechanism directly at the level of the original colleges.

Fix submitted rank-order lists over colleges.  In round $r$, each unassigned student applies to her $r$-th ranked college.  Each college considers the applications received in that round and
processes all of its currently vacant slots in their precedence order.  Each vacant slot permanently
accepts its most-preferred acceptable applicant among the round-$r$ applicants not already
accepted by an earlier vacant slot of the same college.  Applicants not accepted in round $r$
remain unassigned and proceed to the next round.  All acceptances are final.  For a precedence
order $w$ of $c_0$, let $\bmu^{\mathrm{BM}}(w)$ denote the resulting matching and let
$\mu^{\mathrm{BM}}(w)$ denote the set of students assigned to $c_0$.

At the focal college $c_0$, we impose the \emph{criterion-independent acceptability condition}. We have criterion-independent acceptability if for every student $i$ and any two criteria $t,t'$ used by $c_0$, $i \succ_t \emptyset$ if and only if $i \succ_{t'} \emptyset$. The criteria may rank acceptable students differently, but they agree on which students are
acceptable.  The universal-acceptability assumption used in the random-market analysis of
Section~\ref{sec:large_market} is a special case.

For $h\in\{0,\ldots,q\}$, write
\begin{equation}
    w^{-h}\coloneqq (w(h+1),\ldots,w(q))
\end{equation}
be the suffix obtained by deleting the first $h$ slots.  As in
Section~\ref{sec:large_market}, $\Ch_{w^{-h}}(X)$ denotes the set selected by
processing the slots in $w^{-h}$ sequentially from $X$.

The following lemma extends the single-college comparison to a situation in which the first $h$
seats have already been filled by the same set under the two precedence orders.

\begin{lem}\label{lem:bm_truncated_order}
Let $P$ and $X$ be disjoint sets of students such that $|P|=h$ and every student in $X$ is acceptable
to every criterion of $c_0$.  Then
\begin{equation}
\begin{split}
    P\cup\Ch_{v_{BA}^{-h}}(X)
    &\succeq_A
    P\cup\Ch_{v_{AB}^{-h}}(X),\\
    P\cup\Ch_{v_{AB}^{-h}}(X)
    &\succeq_B
    P\cup\Ch_{v_{BA}^{-h}}(X).
\end{split}
\end{equation}
\end{lem}

The intuition is that the common set $P$ can be treated as fixed: by
responsiveness, comparing two completions of $P$ is equivalent to comparing
the newly selected students according to the same underlying criterion
ranking.  
If the first $h$ positions lie within the common prefix, the
truncated orders retain a common prefix followed by the same two blocks, so
the original final-block comparison applies directly.  
If truncation reaches
the $A$- and $B$-blocks, deleting the same number of initial slots leaves
weakly more terminal $A$-slots under $v_{BA}^{-h}$ and weakly more terminal
$B$-slots under $v_{AB}^{-h}$.  For criterion $A$, one can transform the
remaining order under $v_{AB}^{-h}$ into that under $v_{BA}^{-h}$ by first
placing the remaining $A$-block after the $B$-block and then relabeling the
$B$-slots immediately preceding the terminal $A$-block as $A$-slots.  The
first operation benefits $A$ by the final-block swap result, and each
relabeling benefits $A$ because all subsequent slots are already assigned to
$A$.  The argument for criterion $B$ is symmetric, and adding the common set
$P$ back preserves both comparisons.

\begin{thm}
\label{thm:bm_later_is_advantaged}
Fix arbitrary submitted rank-order lists of students and arbitrary but fixed precedence orders and
criterion rankings at colleges other than $c_0$.  Under
the criterion-independent acceptability condition,
\begin{equation}
    \mu^{\mathrm{BM}}(v_{BA})
    \succeq_A
    \mu^{\mathrm{BM}}(v_{AB})
    \qquad\text{and}\qquad
    \mu^{\mathrm{BM}}(v_{AB})
    \succeq_B
    \mu^{\mathrm{BM}}(v_{BA}).
\end{equation}
Thus, under the Boston mechanism, the criterion placed later is advantaged in every finite
market.
\end{thm}

Theorem~\ref{thm:bm_later_is_advantaged} sharpens the interpretation of the preceding
large-market result.  Competition across multiple colleges is not, by itself, what overturns the
single-college comparison.  Under DA in a finite market, a rejection by $c_0$
can initiate a chain that propagates through the market and returns to $c_0$ before its assignment becomes final.  In the large-market limit, this return event becomes unlikely.  Under the Boston mechanism, by contrast, the
first order-sensitive rejection occurs only when $c_0$ simultaneously fills its remaining seats,
so feedback is shut down by the mechanism even in a finite market.

Without criterion-independent acceptability, Theorem~\ref{thm:bm_later_is_advantaged} can fail.  If different criteria disagree on
which students are acceptable, an order may leave a seat vacant in an early round while the other
order fills it.  Later applicants can then change the comparison even though acceptances are
permanent.

\begin{example}\label{exa:bm_unacceptability}
Let $I=\{i_1,i_2,i_3,i_4\}$, $q=2$, $T=\{A,B\}$, and
$C=\{c_0,c_1\}$, where college $c_1$ has capacity one.  The submitted
rank-order lists of the students and the relevant priority orders are as
follows:
\begin{equation}
    \begin{array}{@{}l @{\qquad\qquad} l@{}}
    \begin{aligned}
        \succ_{i_1}: & \; c_0\\
        \succ_{i_2}: & \; c_0\\
        \succ_{i_3}: & \; c_1,c_0\\
        \succ_{i_4}: & \; c_1
    \end{aligned}
    &
    \begin{aligned}
        \succ_A:     & \; i_1,i_3,i_2,i_4,\emptyset\\
        \succ_B:     & \; i_1,i_3,\emptyset,i_2,i_4\\
        \succ_{c_1}: & \; i_4,i_3,\emptyset.
    \end{aligned}
    \end{array}
\end{equation}
Thus, $i_2$ is acceptable to criterion $A$ but unacceptable to criterion
$B$.  In round 1, students $i_1$ and $i_2$ apply to $c_0$, while students
$i_3$ and $i_4$ apply to $c_1$.  College $c_1$ accepts $i_4$ and rejects
$i_3$.

Under $(A,B)$, criterion $A$ accepts $i_1$, while criterion $B$ rejects
$i_2$.  Student $i_3$ then applies to $c_0$ in round 2 and is accepted by
the vacant slot assigned to criterion $B$.  Hence, $\mu^{\mathrm{BM}}(A,B)=\{i_1,i_3\}$.
Under $(B,A)$, criterion $B$ accepts $i_1$ and criterion $A$ accepts
$i_2$ in round 1.  Thus, $\mu^{\mathrm{BM}}(B,A)=\{i_1,i_2\}$.
Since $i_3\succ_A i_2$, we have $\mu^{\mathrm{BM}}(A,B) \succ_A \mu^{\mathrm{BM}}(B,A)$.
Therefore, criterion $A$ is strictly better off when it is placed earlier.
This failure does not disappear merely by enlarging the market: the example
can be embedded in arbitrarily large markets by adding students and colleges
that do not interact with the displayed agents.
\end{example}

Taken together, the DA and Boston results identify finality as the central
mechanism-design distinction.  Under DA, the later-criterion comparison can
fail in a finite market, but is recovered with high probability when returning rejection chains
vanish in a large market.  Under the Boston mechanism, the same comparison holds exactly in a
finite market under criterion-independent acceptability, because the focal college cannot revise
its class after its order first becomes consequential.

\section{Conclusion}

This paper examines how the sequence in which multiple admission criteria are
applied shapes the composition of the admitted class. For a fixed applicant
set, reversing two terminal criterion blocks weakly benefits each criterion
when its own block is processed later, without requiring any alignment across
criterion rankings. Earlier blocks tend to absorb applicants favored by
several criteria, leaving the later block to choose from the residual pool and
thereby exert greater influence over the class at the margin. This
later-block advantage is specific to terminal reversals. With three or more
criteria, an upstream swap also changes the applicants available to downstream
criteria, so either of the swapped criteria may benefit.

Embedding the same choice rule in a matching market isolates the force that
can overturn this later-block advantage: endogenous feedback into the
college's applicant pool. Under student-proposing DA, an order-induced
rejection can travel through other colleges and return to the focal college
with a new applicant, allowing the criterion moved earlier to benefit even
with only two criteria. If such feedback is absent, the relevant applicant
sets coincide, and the single-college comparison applies. Accordingly, in
regular large markets, the expected share of colleges at which the comparison
fails vanishes, and sufficient thickness makes the result uniform across
colleges. Under the Boston mechanism, holding submitted rank-order lists
fixed and imposing criterion-independent acceptability, the later-block
advantage instead holds exactly in every finite market because final
acceptances prevent subsequent feedback from revising the focal college's
class. The central mechanism-design distinction is therefore whether an
order-induced rejection can return before the college's admissions become
final.

Taken together, these results yield a direct design lesson. In the
environments covered by our positive results, a criterion's ranking is more
strongly reflected in the admitted class when its block is placed later. A
criterion applied early may use its capacity on applicants who are also
favored by other criteria, whereas a criterion applied later can use its
capacity to distinguish among the remaining applicants and thus exert greater
influence over the marginal composition of the class. The familiar practice
of placing the criterion deemed most important first can therefore be
counterproductive. If an admissions authority wants a mission-specific or
otherwise salient criterion to shape who is admitted at the margin, assigning
that criterion the later block may better serve that objective. This lesson
concerns terminal-block ordering and is strongest when consequential feedback
is absent or negligible; it does not extend to arbitrary upstream
rearrangements. The key design question is not which criterion is considered
first, but which criterion is positioned to shape the class at the margin.

\bibliographystyle{ecta}
\bibliography{references}

\appendix

\section{Proofs}\label{sec:proofs}

\subsection{Proof of Theorem~\ref{thm:category_slot_increasing}}

\begin{proof}
Fix an arbitrary applicant set $X\subset I$.
Let $P\subset I$ be the common set of students selected before slot $k^*$ under $v$ and $v'$, let $R=X\setminus P$, and put
$r=q-k^*$. Under $v'$, the suffix beginning $k^*$ selects
\begin{equation}
    U\coloneqq\best_t(R,r+1).
\end{equation}
Let $i^*\coloneqq\topacc_{v(k^*)}(R)$. If $i^*=\emptyset$, under $v$, the suffix beginning at slot $k^*$
selects $W=\best_t(R,r)$. Thus $W\subset U$, and every member of
$U\setminus W$ is acceptable to $t$. Responsiveness implies
$P\cup U\succeq_t P\cup W$.

Suppose next that $i^*\neq\emptyset$. The suffix under $v$ selects
\begin{equation}
    W=\{i^*\}\cup\best_t(R\setminus\{i^*\},r).
\end{equation}
If $i^*\succ_t\emptyset$, the sets $U$ and $W$ have the same cardinality, and
$U$ consists of the highest-ranked $t$-acceptable students of that cardinality. Hence $U$ is
obtained from $W$ through a finite sequence of $t$-improving replacements.

If instead $\emptyset\succ_t i^*$, let $V=\best_t(R,r)$. Then
$W=V\cup\{i^*\}$. Either $U=V$, in which case $U$ is obtained by deleting the
$t$-unacceptable student $i^*$, or $U=V\cup\{j\}$ for a $t$-acceptable student
$j$, in which case $U$ is obtained by replacing $i^*$ with $j$. Responsiveness therefore gives
$P\cup U\succeq_t P\cup W$ in every case. Since
$\Ch_{v'}(X)=P\cup U$ and $\Ch_v(X)=P\cup W$, the result follows.
\end{proof}

\subsection{Proof of Corollary~\ref{cor:repeatedly_apply}}

\begin{proof}
Fix $X\subset I$.  
The case $m=0$ is immediate because $v^0=v$.
For $m\in\{1,\ldots,q\}$, we proceed by induction.
For $m=1$, the precedence order $v^1$ is obtained
from $v$ by assigning criterion $t$ to slot $q$.  Because there is no slot
after $q$, the condition in Theorem~\ref{thm:category_slot_increasing} is
vacuously satisfied.  The theorem therefore gives
\[
    \Ch_{v^1}(X)\succeq_t\Ch_v(X).
\]

Now let $m\geq2$ and suppose that
$\Ch_{v^{m-1}}(X)\succeq_t\Ch_v(X)$.  The orders $v^{m-1}$ and $v^m$
differ only at slot
\[
    k^*=q-m+1,
\]
which is assigned criterion $t$ under $v^m$.  Moreover, by the definition
of $v^{m-1}$, every slot $k>k^*$ is already assigned criterion $t$.
Applying Theorem~\ref{thm:category_slot_increasing} to $v^{m-1}$ thus yields
\[
    \Ch_{v^m}(X)\succeq_t\Ch_{v^{m-1}}(X).
\]
Transitivity of $\succeq_t$, together with the induction hypothesis, implies
\[
    \Ch_{v^m}(X)\succeq_t\Ch_v(X).
\]
This completes the induction.
\end{proof}

\subsection{Proof of Theorem~\ref{thm:category_utility_increasing}}

\begin{proof}
Fix an arbitrary applicant set $X\subset I$.
If $q_A=0$ or $q_B=0$, the two precedence orders coincide and both comparisons are immediate.
Hence suppose $q_A,q_B\geq1$.
We prove $\Ch_{v_{AB}}(X) \succeq_B \Ch_{v_{BA}}(X)$; the other inequality follows by symmetry. For notational simplicity, throughout this proof we renumber the slots $\bar q+1,\bar q+2,\dots$ as the first slot, the second slot, etc. We use the acceptable-student operator $\best_t$ defined in the Model section.

Fix integers $k\ge 1$, $a'\ge 1$, and $b'\ge 1$, and consider a precedence order of the form
\begin{equation}
w = \bigl(v,\underbrace{B,\ldots,B}_{k-1},\underbrace{A,\ldots,A}_{a'},\underbrace{B,\ldots,B}_{b'}\bigr). \tag{\ref{eq:wkm1_form}, revisited}
\end{equation}
Let $R_k(w)$ denote the set of remaining applicants just before the $k$-th slot is processed under $w$, and define
\begin{equation}
\begin{split}
    A_k(w)&\coloneqq \best_A(R_k(w),a'),\\
    B_k(w)&\coloneqq \best_B(R_k(w)\setminus A_k(w),b'),\\
    x_k(w)&\coloneqq \topacc_B(R_k(w)),\\
    y_k(w)&\coloneqq \topacc_A(R_k(w)\setminus A_k(w)).
\end{split}
\end{equation}
The values $x_k(w)$ and $y_k(w)$ may equal $\emptyset$; this records a vacant
$B$-slot or the absence of an additional $A$-acceptable applicant, respectively.

The precedence order $\swap(w,k)$ swaps the boundary labels at the
$k$-th and $(k+a')$-th slots. That is, it changes the label of the
$k$-th post-$v$ slot from $A$ to $B$ and the label of the first slot
in the later $B$-block from $B$ to $A$. The precedence order
$\relab(w,k)$ changes only the label of the $k$-th slot from $A$ to $B$.

\begin{condition}[Boundary Invariance]\label{condition:boundary_invariance}
We say that the \emph{boundary invariance condition} is satisfied for precedence order $w$ and the $k$-th slot if the following holds.
If $x_k(w)\notin A_k(w)$, the condition holds automatically. If
$x_k(w)\in A_k(w)$, it requires
\begin{equation}\label{eq:boundary_invariance}
\left[
\begin{array}{ll}
y_k(w)\neq\emptyset \text{ and } y_k(w)\in B_k(w), &\text{or}\\
y_k(w)=\emptyset \text{ and } |B_k(w)|<b'.
\end{array}
\right.
\end{equation}
\end{condition}

\begin{lem}\label{lem:boundary_invariance}
The following statements hold.
\begin{enumerate}
    \item\label{part:swap_boundary_invariance} The swap is matching-invariant, i.e., $\Ch_{\swap(w, k)}(X)=\Ch_w(X)$ if and only if the boundary invariance condition is satisfied for precedence order $w$ and the $k$-th slot.
    \item\label{part:relabel_boundary_invariance} If the boundary invariance condition is not satisfied for precedence order $w$ and the $k$-th slot, the relabeling is matching-invariant, i.e., $\Ch_{\relab(w, k)}(X)=\Ch_w(X)$.
\end{enumerate}
\end{lem}

\begin{proof}
Write $R=R_k(w)$, $A=A_k(w)$, $B=B_k(w)$, $x=x_k(w)$, and $y=y_k(w)$.
Under $w$, the suffix beginning at the $k$-th post-$v$ slot selects $A\cup B$. We verify
part~\ref{part:swap_boundary_invariance} case by case.

If $x=\emptyset$, no applicant in $R$ is acceptable to $B$. The first slot under the swap and
all later $B$-slots are vacant, while the $A$-block still selects $A$. Thus the swap is
matching-invariant, and the boundary condition holds automatically because
$x\notin A$.

Suppose $x\neq\emptyset$ and $x\notin A$. Removing $x$ does not change the applicants selected
by the $a'$ $A$-slots. Moreover, $x$ is the $B$-best applicant in $R\setminus A$, so moving its
selection to the first slot and reducing the final $B$-block by one slot leaves its total selection
unchanged. Hence the swap selects $A\cup B$, exactly as $w$ does. Again, the boundary condition
holds automatically.

It remains to consider $x\in A$. If $y\neq\emptyset$, the $A$-block after the swapped first slot
selects
\begin{equation}
    (A\setminus\{x\})\cup\{y\}.
\end{equation}
The final $B$-block has one fewer slot and faces $(R\setminus A)\setminus\{y\}$.
Consequently, the total selected set remains $A\cup B$ if and only if $y\in B$.

If instead $y=\emptyset$, there is no additional $A$-acceptable applicant, so the $A$-block
after the swapped first slot selects $A\setminus\{x\}$. The final $B$-block selects
$\best_B(R\setminus A,b'-1)$. This is equal to $B$ if and only if the original $B$-block was
not full, that is, if and only if $|B|<b'$. These cases establish the equivalence in
part~\ref{part:swap_boundary_invariance}, including all possible vacancies.

For part~\ref{part:relabel_boundary_invariance}, failure of the boundary condition can occur only
when $x\in A$. Under $\relab(w,k)$, the first slot selects $x$, the subsequent $a'-1$
$A$-slots select $A\setminus\{x\}$, and the final $b'$ $B$-slots face $R\setminus A$ and
select $B$. The suffix therefore selects $A\cup B$, proving
$\Ch_{\relab(w,k)}(X)=\Ch_w(X)$.
\end{proof} 

We return to the proof of Theorem~\ref{thm:category_utility_increasing}. 
Define a sequence of precedence orders $\{w_k\}_{k=0}^{q_B}$ inductively as follows.
Let $w_0\coloneqq v_{AB}$.
For each $k\in\{1,\ldots,q_B\}$, if $w_{k-1}(k)=B$, let
$w_j=w_{j-1}$ for all $j\in\{k,\dots,q_B\}$ and stop the procedure. Otherwise,
$w_{k-1}(k)=A$, and $w_{k-1}$ can be written in the form
\eqref{eq:wkm1_form} for some $a'\ge 1$ and $b'\ge 1$. We define
$w_k \coloneqq \swap(w_{k-1}, k)$ if the boundary invariance condition holds for
$w_{k-1}$ at slot $k$, and $w_k \coloneqq \relab(w_{k-1}, k)$ otherwise.
By Lemma~\ref{lem:boundary_invariance}, the admitted set is invariant at each step:
\begin{equation}\label{eq:mu_invariant_each_step}
\Ch_{w_k}(X)=\Ch_{w_{k-1}}(X)\quad \text{for all }k\in\{1,\dots, q_B\}.
\end{equation}
Let $v^*\coloneqq w_{q_B}$. By construction, the first $q_B$ entries after the prefix $v$ in $v^*$ are all $B$.
Hence there exists an integer $q_A^*\in \{0, \dots, q_A\}$ such that
\begin{equation}
v^* =\bigl(v,\underbrace{B,\ldots,B}_{q_B},\underbrace{A,\ldots,A}_{q_A^*},\underbrace{B,\ldots,B}_{q_A-q_A^*}\bigr).\tag{\ref{eq:vstar_form}, revisited}
\end{equation}
From \eqref{eq:mu_invariant_each_step}, we have $\Ch_{v^*}(X)=\Ch_{v_{AB}}(X)$.

The precedence order $v^*$ is obtained from $v_{BA}$ by replacing the last $(q_A-q_A^*)$ slots (which are $A$ under $v_{BA}$) with $B$.
Therefore, by Corollary~\ref{cor:repeatedly_apply}, we have
\begin{equation}\label{eq:vstar_b_better_than_vBA}
\Ch_{v_{AB}}(X) = \Ch_{v^*}(X)\succeq_B \Ch_{v_{BA}}(X)
\end{equation}
as desired.
\end{proof}

\subsection{Proof of Lemma~\ref{lem:proposal_set_representation}}

\begin{proof}
Every student in $\mathcal A(w)$ proposes first to the first slot of $c_0$.  That slot therefore ends
DA holding its most-preferred acceptable student in $\mathcal A(w)$, or vacant if none of the students in $\mathcal{A}(w)$ is acceptable to that slot. Every student in $\mathcal A(w)$ who is not held by the first slot is eventually rejected there and, under the fixed
increasing-index refined preference, then proposes to the second slot. Consequently, the second slot ends
DA holding its most-preferred acceptable student among the applicants not held by the first slot,
or vacant if none is acceptable. Repeating this argument slot by slot is exactly the sequential construction of $\Ch_w(\mathcal A(w))$.
\end{proof}

\subsection{Proof of Proposition~\ref{prop:common_applicant_set}}

\begin{proof}
Let $X=\mathcal A(v_{AB})=\mathcal A(v_{BA})$.  By
Lemma~\ref{lem:proposal_set_representation}, the two sets of admitted students are
$\Ch_{v_{AB}}(X)$ and $\Ch_{v_{BA}}(X)$.  The claim therefore follows
from Theorem~\ref{thm:category_utility_increasing}, applied to the fixed applicant set $X$.
\end{proof}

\subsection{Proof of Lemma~\ref{lem:no_return_inclusion}}\label{subsec:proof_lem_no_return_inclusion}

\begin{proof}
Let the baseline matching be $\bmu=\bmu(w)$, and let
$X=\mathcal A(w)$, $D=\Ch_{w'}(X)$,
$R=\mu(w)\setminus D$, and $N=D\setminus\mu(w)$.  We construct a stable matching under $w'$ in
which every student outside $X$ is assigned to a college that she strictly prefers to $c_0$ (or
finds $c_0$ unacceptable).

Run the auxiliary $(w,w')$-transition process. For every $i\in N$, we have $i\in D\subset X=\mathcal A(w)$ but
$i\notin\mu(w)$. Thus, under $w$, student $i$ proposed to the slots of
$c_0$ but was ultimately rejected by every slot of $c_0$. Therefore, $i$
strictly prefers $c_0$ to her assignment under $\bmu$. Note that only a student in $N$ who was
matched outside $c_0$ has a placeholder at an old slot; an unmatched student has no old slot
and no placeholder. Suppose that the process does not return, and let $\eta$ be the real-student
assignment obtained after the displacement chain initiated by each student in
$R$ has ended and all remaining placeholders have been deleted. We next
construct from $\eta$ a matching that is stable under $w'$. Step~1 shows that
no pair $(i,s)$ with $s$ occupied under $\eta$ blocks $\eta$ under the
preferences and slot priorities associated with $w'$. Step~2 starts from
$\eta$ and successively fills vacant slots that belong to blocking pairs; the
resulting matching $\nu$ has no blocking pair under $w'$.

\noindent\emph{Step 1: Ruling out blocking pairs involving occupied slots.}

\noindent\emph{Case 1: $s$ is a slot outside $c_0$.}
Throughout
the transition process, a real holder or placeholder is replaced only by a student whom $s$ ranks above it. Thus the final real holder of $s$ is weakly
preferred by $s$ to its holder under $\bmu$ (where a vacancy under $\bmu$ is
represented by $\emptyset$).

\noindent\emph{Case 1(a): $i$ never becomes active.}
Suppose that $s\succ_i'\eta(i)$. If $i\in N$, then $i\in D$, so $i$ is
assigned to a slot of $c_0$ in the initial state of the transition process.
Because the process does not return to $c_0$, this assignment is unchanged,
and hence $\eta(i)$ is a slot of $c_0$. On the other hand, since
$i\in D\subset X=\mathcal A(w)$ and $i\notin\mu(w)$, under $w$ she
proposed to every slot of $c_0$ and was ultimately rejected by all of them.
Therefore, $\eta(i)\succ_i'\bmu(i)$. If $i\in\mu(w)\cap D$, both
$\eta(i)$ and $\bmu(i)$ are slots of $c_0$; because all slots of a college
are consecutive in $\succ_i'$, the relation $s\succ_i'\eta(i)$ implies
$s\succ_i'\bmu(i)$. Finally, suppose that $i\notin D$. Every student in
$R=\mu(w)\setminus D$ becomes active when her displacement chain is
initiated, so the assumption that $i$ never becomes active implies
$i\notin R$. The initialization changes the assignment of a student only if
she belongs to $D$ or $R$, and hence $\eta(i)=\bmu(i)$. Thus, in every case,
$s\succ_i'\bmu(i)$. The
stability of $\bmu$ then implies $\bmu(s)\succ_s i$. Since the holder of $s$
can only improve in $s$'s preference during the transition process,
$\eta(s)\succeq_s\bmu(s)\succ_s i$. Hence $(i,s)$ is not a blocking pair
for $\eta$.

\noindent\emph{Case 1(b): $i$ becomes active.}
There are two ways in which a student can become active for the first time.
If $i\in R$, she is released from $c_0$ and first proposes to the
highest-ranked acceptable slot below the entire block of $c_0$. Otherwise,
$i$ becomes active when she is displaced from a slot $s_0$ outside $c_0$ and
first proposes to the highest-ranked acceptable slot below $s_0$. Because
this is the first time that $i$ becomes active, $s_0$ is her assignment under
$\bmu$. Each time $i$ subsequently becomes active, she resumes immediately
below the slot from which she has just been displaced. Hence her proposal
position moves only downward in her refined preference order.

Consider first a slot $s$ outside $c_0$ that $i$ ranks above the
slot to which she first proposes upon becoming active. If $i\in R$, then her
baseline assignment is a slot of $c_0$. If $i$ ranks $s$ above that assignment,
baseline stability gives $\bmu(s)\succ_s i$; if she ranks $s$ below the block
of $c_0$, then the definition of her first proposal implies that she finds
$s$ unacceptable. In either case, $(i,s)$ cannot block. Now suppose instead
that $i$ first becomes active when she is displaced from her baseline
assignment $s_0$. If $s\succ_i's_0$, baseline stability gives
$\bmu(s)\succ_s i$. If $s=s_0$, the student who displaced $i$ has higher
priority than $i$ at $s$, and the final occupant of $s$ has still higher
priority. Finally, if $s_0\succ_i's$ and $s$ is ranked above the first slot to
which $i$ proposes, then $i$ finds $s$ unacceptable. Thus, in all cases, such
a slot $s$ cannot form a blocking pair with $i$. Here we also use the fact
that the occupant of each slot outside $c_0$ can only be replaced by a
higher-priority student, so its final occupant is weakly preferred by the slot
to its occupant under $\bmu$. Next,
every slot outside $c_0$ at or below the first slot to which $i$ proposes
that she prefers to $\eta(i)$ is reached while $i$ is active in a
displacement chain. At such a slot, $i$ is
either unacceptable, rejected in favor of a higher-priority occupant, or
initially held and later displaced by a higher-priority student. Thereafter,
the occupant could only be replaced by a still higher-priority student. Hence
the final occupant of the slot is ranked above $i$. Thus, no slot outside
$c_0$ that is occupied under $\eta$ is part of a blocking pair.

\noindent\emph{Case 2: $s$ is a slot of $c_0$.}
Write $s=(c_0,k)$ and
$d=\eta(s)$. Because the transition process does not return to $c_0$, the
assignment of students in $D=\Ch_{w'}(X)$ to the slots of $c_0$ remains
unchanged throughout the process. Thus, at step $k$ in the construction of
$D$, slot $s$ selected student $d$.

\noindent\emph{Case 2(a): $i\in X$.}
Consider any student $i\in X$ such that $s\succ_i'\eta(i)$. Student $i$
cannot have been selected by a slot of $c_0$ preceding $s$, because every
student ranks the slots of $c_0$ in increasing order of their indices.
Therefore, $i$ was still available when $s$ selected a student at step $k$.
Since, at step $k$ in the construction of $D=\Ch_{w'}(X)$, slot $s$
selected $d$ rather than $i$, we have
$d\succ_s i$. Hence $(i,s)$ is not a blocking pair for $\eta$.

\noindent\emph{Case 2(b): $i\notin X$.}
Student $i$ never proposed to
the first slot of $c_0$ under the baseline DA. Hence either $i$ finds $c_0$
unacceptable or her assignment $\bmu(i)$ is ranked above every slot of
$c_0$ in her refined preference order. If $i$ never becomes active in the
transition process, her assignment is unchanged, so she does not prefer $s$
to $\eta(i)$. Suppose instead that $i$ becomes active. Her proposal position
then moves downward in her refined preference order from the slot that
released her. 
If she finds $c_0$ acceptable, then, were she to pass below the block
of $c_0$ or become unmatched, she would first have to reach the first
slot of $c_0$.
By definition, this would make the transition process return to
$c_0$. Under the no-return hypothesis, she must therefore be held at a slot
that she ranks above every slot of $c_0$.
If instead $i$ finds $c_0$ unacceptable, individual
rationality of $\eta$ gives
\(\eta(i)\succeq_i'\emptyset\succ_i'(c_0,k)\) for every
\(k\in\{1,\ldots,q\}\). Here $\eta$ is individually rational because every
student assigned to $c_0$ under $\eta$ belongs to $D\subset X$ and is
acceptable to the selecting slot, while every assignment outside $c_0$ is
either inherited from the individually rational matching $\bmu$ or created
through a mutually acceptable proposal. Thus, regardless
of whether $i$ becomes active,
\begin{equation}\label{eq:eta_outside_X_above_c0}
    \eta(i)\succ_i'(c_0,k)
    \qquad
    \text{for every }i\in I\setminus X
    \text{ and every }k\in\{1,\ldots,q\}.
\end{equation}
This conclusion does not depend on whether a particular slot of $c_0$ is
occupied. In particular, $i$ does not prefer $s$ to $\eta(i)$, and $(i,s)$
is not a blocking pair for $\eta$.

\noindent\emph{Step 2: Eliminating blocking pairs involving vacant slots.}
Starting from $\eta$, we fill vacant slots while ensuring that,
after every move, no slot occupied in the current matching belongs to a
blocking pair. Step~1 shows that this property holds before the first move.
If a slot $s$ that is vacant in the current matching belongs to a
blocking pair, assign to $s$ the highest-priority acceptable student among all
students who strictly prefer $s$ to their current assignments, and vacate that
student's former slot, if any. This move preserves feasibility, strictly
improves the moving student, and leaves every other student's assignment
unchanged.

We now show that this property continues to hold after each move.
The newly occupied slot $s$ does not
belong to a blocking pair because its new holder has the highest priority
among all acceptable students who prefer $s$ to their current assignments.
For every other slot that remains occupied, its holder is unchanged, while no
student's assignment becomes worse; hence no new student can form a blocking
pair with that slot. If the moving student's former slot becomes vacant, the
property just established concerns only occupied slots, so it
does not apply to the newly vacant slot. That slot is treated by the same
procedure if it belongs to a blocking pair. Because each move strictly
improves one student's assignment and the sets of students and slots are
finite, the procedure terminates. At termination,
no occupied slot belongs to a blocking pair, and the stopping condition rules out
blocking pairs involving vacant slots.

Because $\eta$ is individually rational, as established in
Case~2(b), and each vacancy-filling move assigns a student to a mutually
acceptable slot that she strictly prefers to her current assignment, every
such move preserves individual rationality. Consequently, the resulting matching $\nu$ is
feasible, individually rational, and has no blocking pair, so it is stable
under $w'$.

Because vacancy filling only improves students,
\begin{equation}
    \nu(i)\succeq_i'\eta(i)
    \qquad\text{for every }i\in I.
\end{equation}
Combining this inequality with \eqref{eq:eta_outside_X_above_c0} gives
\begin{equation}\label{eq:nu_outside_X_above_c0}
    \nu(i)\succeq_i'\eta(i)\succ_i'(c_0,k)
    \qquad
    \text{for every }i\in I\setminus X
    \text{ and every }k\in\{1,\ldots,q\}.
\end{equation}

The matching $\bmu(w')$ is student-optimal among all stable matchings under $w'$.  Hence every
student weakly prefers $\bmu(w')$ to $\nu$. In particular, for every
$i\in I\setminus X$ and every $k\in\{1,\ldots,q\}$,
\begin{equation}
    \bmu(w')(i)\succeq_i'\nu(i)\succ_i'(c_0,k).
\end{equation}
Therefore, no student in $I\setminus X$ proposes to $c_0$ in DA under $w'$. Thus
$\mathcal A(w')\subset X=\mathcal A(w)$.
\end{proof}

\subsection{Proof of Theorem~\ref{thm:no_feedback_order_comparison}}

\begin{proof}
By Lemma~\ref{lem:no_return_inclusion}, no return from
$v_{AB}$ to $v_{BA}$ implies
$\mathcal A(v_{BA})\subset\mathcal A(v_{AB})$.
Applying the lemma in the reverse direction gives the reverse
inclusion. Hence the two applicant sets coincide, and the result
follows from Proposition~\ref{prop:common_applicant_set}.
\end{proof}

\subsection{Proof of Theorem~\ref{thm:large_market_order_comparison}}
\label{subsec:proof_large_market_order_comparison}

Throughout this proof, probabilities and expectations refer to the random
college lists in market $n$. Because the sets of students and colleges and the
list length are finite, the underlying probability space is finite; we
therefore suppress its explicit specification.

We generate students' lists independently, each by sequential draws without
replacement, and expose a new original college only when the algorithm requires
it. If student $i$ has already drawn the distinct set $S_i$, then her next
original-college draw is $c\notin S_i$ with conditional probability
\begin{equation}\label{eq:deferred_draw_probability}
    \frac{p_c^n}{1-\sum_{d\in S_i}p_d^n}.
\end{equation}
Both the baseline DA exploration and the transition explorations below are
non-anticipating: the decision whether to expose another draw, and whose draw
to expose, depends only on information already revealed. Hence
\eqref{eq:deferred_draw_probability} remains valid conditional on every partial
history generated by these explorations.

For a baseline precedence order $w$, let $\mathcal H^n(w)$ denote the
information contained in the complete baseline record. Formally, it is the
sigma-field generated by all original-college draws exposed by the end of DA,
the proposal and tentative-holding history, and the resulting matching. Thus,
for every fixed alternative order $w'$, the baseline applicant set, matching,
and the initial holder and placeholder configuration for the transition
process are determined by $\mathcal H^n(w)$. Conditioning on this record does
not reveal any value of an unexposed list entry, so the deferred-decision rule
continues to apply during the subsequent transition exploration.

Index colleges so that
$p_1^n\geq\cdots\geq p_n^n$. For a baseline run under $w$, let
$Y_c^n(w)$ be the number of colleges among $\{1,\ldots,c\}$ that have not been
revealed in any student's list by the end of baseline DA. This count is
determined by the baseline record, and every college it counts has all of its
slots vacant in the baseline matching.

\begin{lem}[Unused-college estimates]\label{lem:unused_college_estimates}
For every regular sequence and every $c>4k$,
\begin{equation}\label{eq:Yc_elementary_lower_bound}
    \mathbb E[Y_c^n(w)]
    \geq
    \frac{c}{2}\exp\left(-\frac{8\bar q n k}{c}\right).
\end{equation}
Moreover, the occupancy estimates of \citet{kojima2009incentives} give
\begin{equation}\label{eq:Y_variance_bounds}
    \operatorname{Var}[Y_c^n(w)]\leq\mathbb E[Y_c^n(w)],
    \qquad
    \operatorname{Var}[Y_T^n]\leq\mathbb E[Y_T^n].
\end{equation}
These bounds are uniform over baseline precedence orders and slot priorities.
Consequently, for $Y\in\{Y_c^n(w),Y_T^n\}$,
\begin{equation}\label{eq:Y_half_mean_bound}
    \Pr\left[Y\leq\frac{\mathbb E[Y]}{2}\right]
    \leq\frac{4}{\mathbb E[Y]}.
\end{equation}
\end{lem}

\begin{proof}
Let
\begin{equation*}
    Q^n\coloneqq\sum_{d=1}^k p_d^n.
\end{equation*}
Fix a college $d> 2k$. Before $d$ has appeared in a student's list, the total
probability mass of that student's previously drawn colleges is at most $Q^n$.
Therefore, whenever another draw is requested, the conditional probability of
drawing $d$ is at most $p_d^n/(1-Q^n)$. Because the $d-k$ colleges
$k+1,\ldots,d$ are all at least as popular as $d$,
\begin{equation}
    p_d^n\leq\frac{1-Q^n}{d-k}.
\end{equation}
There are at most $|I^n|k\leq\bar q n k$ draws even if every entry of every
student's list is exposed, and absence from all of these draws implies that
$d$ is not exposed by baseline DA. Hence, if $E_d$ denotes the latter event,
\begin{equation}
\begin{split}
    \Pr[E_d]
    &\geq
    \left(1-\frac{1}{d-k}\right)^{\bar q n k}\\
    &\geq
    \exp\left(-\frac{2\bar q n k}{d-k}\right)
    \geq
    \exp\left(-\frac{4\bar q n k}{d}\right).
\end{split}
\end{equation}
If $c>4k$, then every
$d\in\{\lceil c/2\rceil,\ldots,c\}$ satisfies $d>2k$, and therefore
\begin{equation*}
\begin{split}
    \mathbb E[Y_c^n(w)]
    &\geq
    \sum_{d=\lceil c/2\rceil}^c \Pr[E_d]\\
    &\geq
    \frac{c}{2}
    \exp\left(-\frac{8\bar q n k}{c}\right).
\end{split}
\end{equation*}
This proves \eqref{eq:Yc_elementary_lower_bound}.

The variance inequalities in \eqref{eq:Y_variance_bounds} are the
negative-occupancy bounds used in the online appendix of
\citet{kojima2009incentives}. They depend only on the original-college draws,
so introducing internal slots or changing their priorities does not affect
them. Finally, Chebyshev's inequality gives
\begin{equation*}
    \Pr\left[Y\leq\frac{\mathbb E[Y]}{2}\right]
    \leq
    \Pr\left[|Y-\mathbb E[Y]|\geq\frac{\mathbb E[Y]}{2}\right]
    \leq
    \frac{4\operatorname{Var}[Y]}{\mathbb E[Y]^2}
    \leq
    \frac{4}{\mathbb E[Y]},
\end{equation*}
which proves \eqref{eq:Y_half_mean_bound}.
\end{proof}

Define
\begin{equation}
    c^*(n)
    \coloneqq
    \frac{16\bar q n k}{\log(\bar q n)}.
\end{equation}
For all sufficiently large $n$, $c^*(n)>4k$. By
\eqref{eq:Yc_elementary_lower_bound}, for every $c>c^*(n)$ and every baseline
order $w$,
\begin{equation}\label{eq:Yc_lower_bound}
    \mathbb E[Y_c^n(w)]
    \geq
    \frac{8k\sqrt{\bar q n}}{\log(\bar q n)}.
\end{equation}
In particular, the right-hand side diverges uniformly over
$c>c^*(n)$ and over baseline orders.

For the sufficiently thick case, recall that $N_c^n$ is the number of
students whose complete list contains $c$. For a baseline run under $w$, let
$\hat N_c^n(w)$ be the number of students for whom $c$ has been exposed by the
end of baseline DA, and define
\begin{equation}
    \hat V_T^n(w)
    \coloneqq
    \left\{c\in C^n:
    \frac{p_{\max}^n}{p_c^n}\leq T
    \text{ and }
    \hat N_c^n(w)<q_c
    \right\},
    \qquad
    \hat Y_T^n(w)\coloneqq|\hat V_T^n(w)|.
\end{equation}
Both objects are determined by the baseline record. Since
$\hat N_c^n(w)\leq N_c^n$ pathwise,
\begin{equation}
    V_T^n\subset\hat V_T^n(w)
    \qquad\text{and}\qquad
    Y_T^n\leq\hat Y_T^n(w).
\end{equation}
Moreover, every $d\in\hat V_T^n(w)$ has a vacant slot in the baseline
matching, because fewer than $q_d$ students have exposed $d$.

We next define fixed-release return events. Conditional on a realized baseline
record, fix an alternative order $w'$ and the holder and placeholder
configuration induced by the transition construction. For every nonempty
$B\subset\mu_c^n(w)$, let $\hat F_{c,B}^n(w,w')$ be the event that the
corresponding transition exploration returns to $c$ when $B$, rather than the
endogenous release set, is used as its initial release set. The members of $B$
are detached from $c$ and processed in the fixed order $\triangleright$; the
assignment within $c$ is ignored until a return occurs. When
$B=R(w,w')$, this is exactly the transition process defined in the main text.
For other $B$, it is only an auxiliary exploration used below in a finite
union bound.

\begin{lem}[Bounded-release return bounds]\label{lem:bounded_release_return}
\begin{enumerate}
    \item For all sufficiently large $n$, every $c>c^*(n)$, and on the event
    \begin{equation*}
        Y_c^n(w)>\frac{\mathbb E[Y_c^n(w)]}{2},
    \end{equation*}
    the following bound holds simultaneously for every nonempty
    $B\subset\mu_c^n(w)$, where $b=|B|\leq\bar q$:
    \begin{equation}\label{eq:pathwise_fixed_release_regular}
        \Pr\left[
            \hat F_{c,B}^n(w,w')
            \,\middle|\,
            \mathcal H^n(w)
        \right]
        \leq
        \frac{4b}{\mathbb E[Y_c^n(w)]}
        \leq
        \frac{4\bar q}{\mathbb E[Y_c^n(w)]}.
    \end{equation}

    \item Suppose the sequence is sufficiently thick for some finite $T$, and
    put $M_T^n\coloneqq\mathbb E[Y_T^n]$. For all sufficiently large $n$,
    every college $c$, and on the event
    \begin{equation*}
        \hat Y_T^n(w)>\frac{M_T^n}{2},
    \end{equation*}
    the following bound holds simultaneously for every nonempty
    $B\subset\mu_c^n(w)$, where $b=|B|\leq\bar q$:
    \begin{equation}\label{eq:pathwise_fixed_release_thick}
        \Pr\left[
            \hat F_{c,B}^n(w,w')
            \,\middle|\,
            \mathcal H^n(w)
        \right]
        \leq
        \frac{4Tb}{M_T^n}
        \leq
        \frac{4T\bar q}{M_T^n}.
    \end{equation}
\end{enumerate}
Both bounds are uniform over slot priorities and over holder and placeholder
configurations generated by the transition construction.
\end{lem}

\begin{proof}
Fix a realized baseline record and a nonempty release set $B$, and write
$X=\mathcal A(w)$. Each member of $B$ initiates one displacement chain. At
every point in a chain there is at most one active student: displacing a real
holder changes the identity of the active student but does not branch the
chain. Moving through several slots of the same original college requires no
additional original-college draw. Because every slot finds every student
acceptable, reaching a college with a vacant slot ends the current chain.
Slot-specific priorities affect only the identity of the next active student.

Consider first part~1. Write
\begin{equation*}
    M_c^n(w)\coloneqq\mathbb E[Y_c^n(w)],
    \qquad
    y\coloneqq Y_c^n(w)>\frac{M_c^n(w)}{2}.
\end{equation*}
The $y$ baseline-unrevealed colleges among $\{1,\ldots,c\}$ are vacant and
have no placeholders. Before the $\ell$-th displacement chain begins, at most
$\ell-1$ of them can have been used as absorbers. Hence at least
$y-(\ell-1)$ remain available.

Fix such a set $U$ of available absorbers. Before the current chain reaches
one of these colleges, every member of $U$ remains unexposed to every carrier.
If the active student belongs to $X$, she cannot return to $c$; if she lies
outside $X$, $c$ has not yet been exposed in her list. Thus, at every partial
history at which a return remains possible, conditional on the next new draw
belonging to $\{c\}\cup U$, the probability that it is $c$ equals
\begin{equation}
    \frac{p_c^n}{p_c^n+\sum_{d\in U}p_d^n}
    \leq
    \frac{1}{|U|+1},
\end{equation}
where the first expression follows from
\eqref{eq:deferred_draw_probability} and the inequality follows because every
$d\in U\subset\{1,\ldots,c\}$ satisfies $p_d^n\geq p_c^n$. Reaching a
member of $U$ absorbs the carrier and ends the chain. Hence the same bound
applies to the probability that the chain returns before it is absorbed. If
no such decisive draw occurs before the relevant lists are exhausted, the
chain does not return.

By \eqref{eq:Yc_lower_bound}, $M_c^n(w)\to\infty$ uniformly over
$c>c^*(n)$. Thus, for all sufficiently large $n$ and every
$\ell\leq b\leq\bar q$,
\begin{equation*}
    |U|+1
    \geq
    y-(\ell-1)+1
    \geq
    \frac{M_c^n(w)}{4}.
\end{equation*}
Each chain therefore returns with conditional probability at most
$4/M_c^n(w)$. A union bound over the $b$ chains proves
\eqref{eq:pathwise_fixed_release_regular}.

For part~2, write
\begin{equation*}
    \hat y\coloneqq\hat Y_T^n(w)>\frac{M_T^n}{2}.
\end{equation*}
Every college in $\hat V_T^n(w)$ has a baseline vacancy and satisfies
\begin{equation*}
    p_d^n\geq\frac{p_{\max}^n}{T}\geq\frac{p_c^n}{T}.
\end{equation*}
Remove the possible focal college $c$ from the absorber set. Before the
$\ell$-th chain, at most $\ell-1$ further absorbers have been used. At any
partial history of that chain, remove additionally the at most $k$ colleges
already appearing in the current carrier's list. The remaining set $U$ of
available absorbers satisfies
\begin{equation*}
    |U|\geq\hat y-\ell-k.
\end{equation*}
Since $M_T^n\to\infty$ under sufficient thickness, for all sufficiently large
$n$ and every $\ell\leq\bar q$ we have $|U|\geq M_T^n/4$. If the current
carrier belongs to $X$, return is impossible at that history. Otherwise,
conditional on the next new draw belonging to $\{c\}\cup U$, the
deferred-decision rule gives
\begin{equation*}
    \Pr[\text{next draw}=c
        \mid\text{next draw}\in\{c\}\cup U]
    =
    \frac{p_c^n}{p_c^n+\sum_{d\in U}p_d^n}
    \leq
    \frac{1}{1+|U|/T}
    \leq
    \frac{4T}{M_T^n}.
\end{equation*}
This bound holds after every possible change of carrier. Since drawing any
member of $U$ absorbs the carrier, the probability that the current chain
returns is at most $4T/M_T^n$. A union bound over the $b$ chains proves
\eqref{eq:pathwise_fixed_release_thick}.

Finally, placeholders do not invalidate either absorber count. A
baseline-unrevealed college has no placeholder. At a college in
$\hat V_T^n(w)$, the inequality $\hat N_d^n(w)<q_d$ leaves at least one vacant
slot even after holders moved to $c$ are replaced by placeholders. Replacing a
placeholder ends the current chain, while rejection by a placeholder merely
preserves the old cutoff. Thus placeholders cannot eliminate an absorber or
increase the return probability.
\end{proof}

\begin{proof}
For precedence orders $w,w'$ of college $c$, let $F_c^n(w,w')$ be the event
that the $(w,w')$-transition process returns to $c$. By
Theorem~\ref{thm:no_feedback_order_comparison},
\begin{equation}\label{eq:failure_implies_feedback}
    (\mathcal G_c^n)^c
    \subset
    F_c^n(v_{A_cB_c}^c,v_{B_cA_c}^c)
    \cup
    F_c^n(v_{B_cA_c}^c,v_{A_cB_c}^c).
\end{equation}
It is therefore enough to bound one directed return event.

Fix a baseline order $w$ and an alternative order $w'$ at college $c$. If
$R(w,w')=\emptyset$, no displacement chain is initiated and the transition
cannot return. Otherwise, the realized release set is a nonempty subset of
$\mu_c^n(w)$, and the corresponding fixed-release exploration coincides with
the actual transition process. Hence, pathwise,
\begin{equation}\label{eq:return_union_over_release_sets}
    F_c^n(w,w')
    \subset
    \bigcup_{\emptyset\neq B\subset\mu_c^n(w)}
    \hat F_{c,B}^n(w,w').
\end{equation}
This finite union lets us avoid relying on any separate conditioning argument
for the endogenous release set.

We first prove part~\ref{part:1_large_market_order_comparison}. Fix
$c>c^*(n)$ and put
\begin{equation*}
    M_c^n(w)\coloneqq\mathbb E[Y_c^n(w)],
    \qquad
    E_c^n(w)
    \coloneqq
    \left\{Y_c^n(w)>\frac{M_c^n(w)}{2}\right\}.
\end{equation*}
On $E_c^n(w)$, Lemma~\ref{lem:bounded_release_return} gives, simultaneously
for every nonempty $B\subset\mu_c^n(w)$,
\begin{equation}\label{eq:fixed_release_return_bound}
    \Pr\left[
        \hat F_{c,B}^n(w,w')
        \,\middle|\,
        \mathcal H^n(w)
    \right]
    \leq
    \frac{4\bar q}{M_c^n(w)}.
\end{equation}
Conditional on the baseline record, $\mu_c^n(w)$ is fixed and has cardinality
at most $q_c\leq\bar q$. Thus it has at most $2^{\bar q}-1$ nonempty subsets.
Combining \eqref{eq:return_union_over_release_sets} and
\eqref{eq:fixed_release_return_bound}, we obtain on $E_c^n(w)$
\begin{equation}\label{eq:conditional_directed_union_bound}
    \Pr\left[
        F_c^n(w,w')
        \,\middle|\,
        \mathcal H^n(w)
    \right]
    \leq
    \frac{4\bar q(2^{\bar q}-1)}{M_c^n(w)}.
\end{equation}
Moreover, Lemma~\ref{lem:unused_college_estimates} gives
\begin{equation}\label{eq:Yc_concentration}
    \Pr[(E_c^n(w))^c]
    \leq
    \frac{4}{M_c^n(w)}.
\end{equation}
Because $E_c^n(w)$ is determined by the baseline record, the tower property
and \eqref{eq:conditional_directed_union_bound} yield
\begin{equation*}
\begin{split}
    \Pr[F_c^n(w,w')]
    &=
    \mathbb E\left[
        \mathbf 1_{E_c^n(w)}
        \Pr[F_c^n(w,w')\mid\mathcal H^n(w)]
    \right]\\
    &\quad+
    \mathbb E\left[
        \mathbf 1_{(E_c^n(w))^c}
        \Pr[F_c^n(w,w')\mid\mathcal H^n(w)]
    \right]\\
    &\leq
    \frac{4\bar q(2^{\bar q}-1)}{M_c^n(w)}
    +\Pr[(E_c^n(w))^c].
\end{split}
\end{equation*}
Using \eqref{eq:Yc_concentration}, we conclude that
\begin{equation}\label{eq:directed_return_bound_union}
    \Pr[F_c^n(w,w')]
    \leq
    \frac{4\left[\bar q(2^{\bar q}-1)+1\right]}{M_c^n(w)}.
\end{equation}

Combining \eqref{eq:directed_return_bound_union} with
\eqref{eq:Yc_lower_bound}, for every $c>c^*(n)$ and every directed comparison
$(w,w')$ considered here,
\begin{equation}\label{eq:epsilon_n_definition}
    \Pr[F_c^n(w,w')]
    \leq
    \varepsilon_n
    \coloneqq
    \frac{
        \left[\bar q(2^{\bar q}-1)+1\right]
        \log(\bar q n)
    }{
        2k\sqrt{\bar q n}
    },
\end{equation}
where $\varepsilon_n\to0$. Applying this bound to both directed transitions in
\eqref{eq:failure_implies_feedback} gives
\begin{equation*}
    \Pr[(\mathcal G_c^n)^c]
    \leq2\varepsilon_n
    \qquad
    \text{for every }c>c^*(n).
\end{equation*}
For the remaining colleges, use the trivial bound of one. Therefore,
\begin{equation*}
\begin{split}
    \frac{1}{n}
    \sum_{c\in C^n}\Pr[(\mathcal G_c^n)^c]
    &\leq
    \frac{\lceil c^*(n)\rceil}{n}
    +2\varepsilon_n
    \longrightarrow0,
\end{split}
\end{equation*}
because $c^*(n)/n\to0$. This proves
part~\ref{part:1_large_market_order_comparison}.

We now prove part~\ref{part:2_large_market_order_comparison}. Suppose the
sequence is sufficiently thick for some finite $T$, and write
\begin{equation*}
    M_T^n\coloneqq\mathbb E[Y_T^n].
\end{equation*}
Fix any college $c$ and any directed comparison $(w,w')$, and define the
baseline-history event
\begin{equation*}
    \hat E_T^n(w)
    \coloneqq
    \left\{\hat Y_T^n(w)>\frac{M_T^n}{2}\right\}.
\end{equation*}
On this event, Lemma~\ref{lem:bounded_release_return} gives, simultaneously
for every nonempty $B\subset\mu_c^n(w)$,
\begin{equation}\label{eq:fixed_release_thickness_bound}
    \Pr\left[
        \hat F_{c,B}^n(w,w')
        \,\middle|\,
        \mathcal H^n(w)
    \right]
    \leq
    \frac{4T\bar q}{M_T^n}.
\end{equation}
The same finite union as above therefore gives, on $\hat E_T^n(w)$,
\begin{equation*}
    \Pr\left[
        F_c^n(w,w')
        \,\middle|\,
        \mathcal H^n(w)
    \right]
    \leq
    \frac{4T\bar q(2^{\bar q}-1)}{M_T^n}.
\end{equation*}
Furthermore, the pathwise inequality
$Y_T^n\leq\hat Y_T^n(w)$ and
Lemma~\ref{lem:unused_college_estimates} imply
\begin{equation}\label{eq:YT_concentration}
\begin{split}
    \Pr[(\hat E_T^n(w))^c]
    &\leq
    \Pr\left[Y_T^n\leq\frac{M_T^n}{2}\right]\\
    &\leq
    \frac{4}{M_T^n}.
\end{split}
\end{equation}
Applying the tower property exactly as above yields
\begin{equation}\label{eq:uniform_directed_return_bound}
    \Pr[F_c^n(w,w')]
    \leq
    \frac{
        4\left[T\bar q(2^{\bar q}-1)+1\right]
    }{
        M_T^n
    }
\end{equation}
for every $c\in C^n$ and either directed comparison $(w,w')$ at that college.
Applying \eqref{eq:uniform_directed_return_bound} to both directions in
\eqref{eq:failure_implies_feedback}, we obtain
\begin{equation*}
    \sup_{c\in C^n}
    \Pr[(\mathcal G_c^n)^c]
    \leq
    \frac{
        8\left[T\bar q(2^{\bar q}-1)+1\right]
    }{
        \mathbb E[Y_T^n]
    }
    \longrightarrow0,
\end{equation*}
because sufficient thickness requires
$\mathbb E[Y_T^n]\to\infty$. This proves
part~\ref{part:2_large_market_order_comparison}.
\end{proof}

\subsection{Proof of Lemma~\ref{lem:bm_truncated_order}}

\begin{proof}
For a criterion $t$ and a fixed set $P$, define a preference over sets disjoint from $P$ by
\begin{equation}
    S\succeq_t^P S'
    \quad\Longleftrightarrow\quad
    P\cup S\succeq_t P\cup S'.
\end{equation}
Because $\succeq_t$ is responsive to $\succ_t$, the induced preference $\succeq_t^P$ is also
responsive to the same ranking $\succ_t$.  We may therefore apply the single-college results using
$\succeq_t^P$ as the relevant responsive extension.

Let $\bar q$ be the length of the common prefix $v$ in $v_{AB}$ and $v_{BA}$.
If $h\leq\bar q$, the two suffixes share
$(v(h+1),\ldots,v(\bar q))$ as a common prefix, followed by the
same $A$- and $B$-blocks in opposite orders.
Applying Theorem~\ref{thm:category_utility_increasing} with the
responsive preferences conditioned on the common set $P$ gives both
desired inequalities.

Suppose next that $h>\bar q$.  The two suffixes can be written as
\begin{equation}
    v_{AB}^{-h}
    = (\underbrace{A,\ldots,A}_{a},\underbrace{B,\ldots,B}_{b}),
    \qquad
    v_{BA}^{-h}
    = (\underbrace{B,\ldots,B}_{b'},\underbrace{A,\ldots,A}_{a'}),
\end{equation}
where
\begin{equation}
    a+b=a'+b'=q-h,
    \qquad
    a'\geq a,
    \qquad
    b\geq b'.
\end{equation}
Empty blocks are allowed. 
For notational convenience, we denote
\begin{equation}
    B^bA^a
    = (\underbrace{B,\ldots,B}_{b},\underbrace{A,\ldots,A}_{a}),
    \qquad
    A^{a'}B^{b'}
    = (\underbrace{A,\ldots,A}_{a'},\underbrace{B,\ldots,B}_{b'}).
\end{equation}

Applying
Theorem~\ref{thm:category_utility_increasing} under the responsive preference
$\succeq_A^P$ gives
\begin{equation}
    P\cup\Ch_{B^bA^a}(X)
    \succeq_A
    P\cup\Ch_{v_{AB}^{-h}}(X).
\end{equation}
Starting from $B^b A^a$, replace the last $b-b'$ $B$-slots, one at a time from right to left, by
$A$-slots.  At every replacement, all subsequent slots are assigned to $A$, so
Theorem~\ref{thm:category_slot_increasing}, again under the conditioned responsive preference,
applies.  The resulting order is $v_{BA}^{-h}$.  Hence
\begin{equation}
    P\cup\Ch_{v_{BA}^{-h}}(X)
    \succeq_A
    P\cup\Ch_{B^bA^a}(X)
    \succeq_A
    P\cup\Ch_{v_{AB}^{-h}}(X).
\end{equation}

For criterion $B$, first apply Theorem~\ref{thm:category_utility_increasing} to compare
$v_{BA}^{-h}$ with $A^{a'}B^{b'}$.  Then replace the last $a'-a$ $A$-slots in
$A^{a'}B^{b'}$, from right to left, by $B$-slots.  All subsequent slots are assigned to $B$, so
Theorem~\ref{thm:category_slot_increasing} applies at every step and produces $v_{AB}^{-h}$.
This yields the second inequality.
\end{proof}

\subsection{Proof of Theorem~\ref{thm:bm_later_is_advantaged}}

\begin{proof}
Couple the two Boston runs using the same submitted rank-order lists.  
If the two runs accept the
same set of students at $c_0$ in every round, their final admitted sets are identical and the result
is immediate.  
Otherwise, let $r^*$ be the first round in which the sets accepted by $c_0$ differ.
Before round $r^*$, the two runs have accepted and rejected the same sets of
students at $c_0$, although their assignments to individual slots may differ.
Consequently, the histories at every other college are identical, and the same set $X$ of
students applies to $c_0$ in round $r^*$.

Let $P$ be the common set of students already accepted by $c_0$, let $h=|P|$, and let $X^+$ be
the subset of $X$ acceptable to the criteria of $c_0$.  This subset is well-defined by criterion-independent acceptability.  Because every student accepted in an earlier
round is acceptable to every criterion of $c_0$, the occupied slots form the first $h$ positions of
each precedence order.  The vacant slots are therefore described by the suffixes
$v_{AB}^{-h}$ and $v_{BA}^{-h}$.  If $|X^+|<q-h$, both precedence orders accept all
students in $X^+$ and reject all students in $X\setminus X^+$.  If $|X^+|=q-h$, both orders accept
exactly $X^+$.  In either case, the accepted sets cannot first differ in round $r^*$.  Therefore, $|X^+|>q-h$.
Both runs fill all remaining seats in round $r^*$, and their final admitted sets are
\begin{equation}
    P\cup\Ch_{v_{AB}^{-h}}(X^+)
    \qquad\text{and}\qquad
    P\cup\Ch_{v_{BA}^{-h}}(X^+),
\end{equation}
respectively.  Lemma~\ref{lem:bm_truncated_order} gives the two desired comparisons.  Because
acceptances are final and $c_0$ is full at the end of round $r^*$, any subsequent rejection cascade
elsewhere in the market cannot change its admitted set.
\end{proof}

\end{document}